\documentclass[runningheads,12pt]{llncs}
\usepackage{graphicx}
\usepackage{amsfonts,amssymb}
\usepackage{mathtools}
\usepackage[usenames,dvipsnames,table]{xcolor}
\usepackage{algorithm}
\usepackage[noend]{algpseudocode}
\usepackage{ifthen}
\usepackage{subcaption} 

\usepackage{etoolbox}
\AtEndEnvironment{proof}{\phantom{}\qed}

\newcommand{\probname}[1]{{\color{blue!80!black}{\textbf{{#1}}}}}
\newtheorem{observation}{Observation}
\usepackage[framemethod=tikz]{mdframed}
\newmdenv[
  backgroundcolor=gray!15,
  topline=false,
  bottomline=false,
  rightline=false,
  skipabove=\topsep,
  skipbelow=\topsep
]{siderules}

\newcommand{\red}[1]{{\textcolor{red}{#1}}}

\newcommand{\myremark}[4]{\textcolor{blue}{\textsc{#1 #2: }}\textcolor{#4}{\textsf{#3}}}

\newcommand{\rodrigo}[2][says]{\myremark{Rodrigo}{#1}{#2}{Plum}}
\newcommand{\maria}[2][says]{\myremark{Maria}{#1}{#2}{Red}}

\newcommand{\vahideh}[2][says]{\myremark{Vahideh}{#1}{#2}{magenta}} 
\newcommand{\kda}{\textit{MCSI-AkD}} 
\newcommand{\mmcsi}{\textit{mCSI}}
\newcommand{\mcsi}{\textit{MCSI}}
\newcommand{\lmcsi}{\textit{L-MCSI}}

 \usepackage[scale=0.82]{geometry}
\begin{document}
\title{Computing largest minimum color-spanning intervals of imprecise points\thanks{A preliminary version of this paper appeared in Proc. 16th Latin American Symposium on Theoretical Informatics, Part~I, pp. 81–96.}}
%
%
\author{Ankush Acharyya\inst{1} \and Vahideh Keikha\inst{2} \and
Maria Saumell\inst{3} \and
Rodrigo I. Silveira\inst{4}}
\authorrunning{A. Acharyya et al.}
%
\institute{Dept. of Computer Science and Engineering, National Institute of Technology, Durgapur, India\\ \email{aacharyya.cse@nitdgp.ac.in} \and The Czech Academy of Sciences, Institute of Computer Science, Czech Republic\\\email{keikha@cs.cas.cz} \and
Dept. of Theoretical Computer Science, Faculty of Information Technology, Czech Technical University in Prague, Czech Republic\\
\email{maria.saumell@fit.cvut.cz}
 \and
Dept. de Matem\`{a}tiques, Universitat Polit\`{e}cnica de Catalunya, Spain \\ \email{rodrigo.silveira@upc.edu}}

\maketitle              
\begin{abstract}
We study a geometric facility location problem under imprecision.
Given $n$ unit intervals in the real line, each with one of $k$ colors, the goal is to place one point in each interval such that the resulting \emph{minimum color-spanning interval} is as large as possible.
A minimum color-spanning interval is an interval of minimum size that contains at least one point from a given interval of each color.
We prove that if the input intervals are pairwise disjoint, the  problem can be solved in
$O(n)$ time, even for intervals of arbitrary length.
For overlapping intervals, the problem becomes much more difficult.
Nevertheless, we show that it can be solved in $O(n \log^2 n)$ time when $k=2$, by exploiting several structural properties of candidate solutions, combined with a number of advanced algorithmic techniques. 
Interestingly, this shows a sharp contrast with the 2-dimensional version of the problem, recently shown to be NP-hard.

\end{abstract}

\keywords{Color-spanning interval \and Imprecise points \and Algorithms} 
%
%
%


\section{Introduction}
Color-spanning problems emerge naturally in certain facility location scenarios, where the objective is to identify an optimal facility location relative to a set of sites, each of which is categorized by a specific attribute (or \emph{color}).
In such settings, the goal is often to find a location that contains at least one site of each attribute and is of optimum size. 
For instance, one may be interested in a location such that the maximum distance to reach one site of each color is as small as possible.
Then, from a geometric point of view, if sites are points in the plane, and the distance used is the Euclidean distance, one is looking for a smallest circle that contains at least one point from each color.
This is known as a  \emph{minimum color-spanning circle}~\cite{abellanas2001smallest}.

Plenty of variants of color-spanning objects have been studied.
Typically, the object sought is a two-dimensional region of some type, such as a circle, a square, or a strip.
Then one can aim at finding the smallest, largest, narrowest, etc., color-spanning object of such a type (see, for example, \cite{abellanas2001smallest,das2009smallest}).
The computation of color-spanning objects for point sites has been vastly studied.

However, it is well-known that the data used in real-world instances is not 100\% accurate.
This is especially true for geometric data, which---in most applications---originates from inaccurate measuring devices, such as GPS receivers or laser scanners.
This motivated a flurry of research on uncertainty models for geometric algorithms, where the imprecision in the data is modeled explicitly.
One of the simplest and most studied models for geometric uncertainty is based on regions: instead of assuming that the exact location of each site is known, one assumes that the site lies within a region (e.g., a disk).
In principle, any location within the site's region is possible.
Choosing one location inside each site's region results in a \emph{realization} of the imprecise sites.
Since many different realizations are possible, natural optimization problems arise.
Most typically, one is interested in understanding extreme realizations: those that give the best possible situation or the worst one.
For instance, in the context of color-spanning circles, if each site is modeled by, say, a disk, then one can wonder how to place one point inside each disk so that the resulting set of points gives the smallest or largest color-spanning circles.

Many problems in computational geometry have been studied for the region-based imprecision model (see, e.g., \cite{loffler2010largest,bbdw}).
Depending on the region (e.g., a line segment, a square, a disk) and the actual problem (e.g., convex hull, triangulation, etc.), some problems become very difficult already under very simple imprecision models, while some others can still be solved efficiently.

In this work, we study the problem of finding a minimum color-spanning circle of the largest possible size for imprecise points modeled as 1D intervals on the real line. 
Our motivation stems from recent work by Acharyya et al.~\cite{acharyya2022minimum}, where the problem was studied for regions consisting of disks in 2D.
While the authors of~\cite{acharyya2022minimum} managed to find efficient algorithms to find a minimum color-spanning circle of \emph{smallest} size, the maximization versions resulted more difficult.
In fact, they proved that the problem of placing one point in each disk, such that the minimum color-spanning circle has the largest possible size, is NP-hard, even for unit input disks and only two colors.
Given this somewhat surprising negative result, in this paper, we study the same problem,  one dimension lower, where disks become intervals on the real line, and the minimum color-spanning circle becomes the \emph{minimum color-spanning interval}. In this case, finding a minimum color-spanning interval of smallest size is trivial (it can be solved in $O(n)$ time following the same strategy as for the analogous problem for circles described in~\cite{acharyya2022minimum}), so we focus on the problem of largest size.  
More formally, the problem we study is defined as follows. 


\probname{Largest minimum color-spanning interval (\lmcsi):} 
Given $n$ unit-length closed intervals, ${\cal I}=\{I_1, I_2,\ldots, I_n\}$, on the real line, each colored with one of $k$ colors,
specified in sorted order with respect to their left endpoints, 
find a realization of $\cal I$ such that the length of the minimum color-spanning interval(s) (\mcsi) of the realization is as large as possible.

See Fig.~\ref{fig:lmcsi_example} for an example with $k=3$.

\begin{figure}[t]
    \centering
        \includegraphics{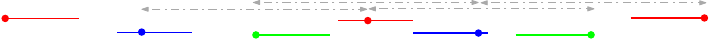}
    \caption{Example of unit intervals with $k=3$ colors, with an optimal realization that results in {four} \lmcsi s (indicated in gray).
    Note that the three leftmost representatives also form a color-spanning interval, which is not minimum.  
     }
    \label{fig:lmcsi_example}
\end{figure}

\paragraph{Contributions}
We first show that if the input intervals are pairwise disjoint, the $\lmcsi$ problem can be solved in
$O(n)$ time, even if intervals have arbitrary lengths.
It turns out that the main difficulty of the problem originates from intervals that overlap.
Intuitively, when two intervals of different color overlap, the points in each interval can be placed in any of two orders.
This, repeated for all pairs of intersecting intervals of different color, results in a combinatorial explosion of possible orderings, already for two colors (i.e., $k=2$).
Thus, most of this work is devoted to efficiently solving the problem for $k=2$.
As is usual in this type of optimization setting, first we focus on the decision version of the problem: Given $n$ colored unit intervals, is there a realization with a minimum color-spanning interval of length at least $q$?
We show that this problem has a rich structure that allows us to compute solutions efficiently.
Indeed, a key contribution is a detailed analysis of the structure of certain canonical (sub)solutions, which we call \emph{leftmost}, and their decomposition into so-called \emph{tabular subsolutions}.
%
By applying a number of advanced algorithmic techniques, we show how to check in $O(n\log n)$ time if some combination of partial tabular subsolutions can be combined into one solution for the whole input,  answering the decision question.

Then we turn our attention to how to efficiently apply the decision procedure to solve the original optimization problem.
To that end, we characterize the candidates for the optimal interval length $\lmcsi$(s). Even though there might be $\Theta(n^3)$ values, using further observations we are able to run a binary search among them that generates only $O(\log n)$ values and solves the $\lmcsi$ problem in $O(n \log^2 n)$ time.

\paragraph{Related work}

An extensive literature exists devoted to color-spanning objects. Basic geometric objects such as the smallest color-spanning circle or square can be computed in $O(nk\log n)$ time~\cite{abellanas2001smallest}. There also exist polynomial-time algorithms to compute other optimal color-spanning objects such as the narrowest strip~\cite{das2009smallest}, the smallest axis-aligned rectangle~\cite{das2009smallest} or square~\cite{khanteimouri2013computing}, the smallest arbitrarily oriented rectangle~\cite{das2009smallest}, the smallest equilateral triangle of fixed orientation~\cite{hasheminejad2015computing}, and the narrowest annulus with the shape of a circle, an axis-parallel square or rectangle, or of an equilateral triangle of fixed orientation~\cite{acharyya2018minimum}. In contrast, the problem of identifying a set of $k$ points of disjoint colors such that the chosen points have the smallest possible diameter is NP-hard~\cite{FleischerX11}.

There are also a few results for color-spanning objects in one dimension. In this setting, the input is usually a set of $n$ points of $k$ colors and an integer $s_i$ associated to each color $i$. The goal is to find $t$ intervals covering at least $s_i$
 points of each color $i$ and such that the maximum length of the intervals is minimized. If the input points are sorted and $t=1$, the problem can be solved in $O(n)$ time~\cite{chen2013algorithms}. For $t=2$ and $s_i=1$ for all $i$, the problem can be solved in $O(n^2 \log n)$ time~\cite{khanteimouri2013spanning}. The running time for this variant has been improved to $O(n^2)$ even for arbitrary $s_i \ge 1$~\cite{jiang2014shortest}. The authors also proved that, for arbitrary values of $t$, it is NP-hard to compute an approximation of the solution within a factor of $1+\varepsilon$.  Further results on the general problem as well as its parameterized complexity are described in~\cite{DBLP:journals/dam/BanerjeeMN20}.

Several strategies have been proposed in the literature to deal with imprecision in geometric data.
Here, we only mention a few relevant results for the region-based model, which is the one adopted in this article.
In this model, finding a placement of points within a set of disks that maximizes or minimizes the radius of the smallest enclosing circle of the points can be solved in $O(n)$ time~\cite{bbdw}. Finding a placement of points within a set of line segments or squares that maximizes or minimizes the area or the perimeter of the convex hull can be solved with algorithms with running times ranging from $O(n)$  to $O(n^{13})$, while some variants are NP-hard~\cite{loffler2010largest}. 
Other objective functions and/or regions lead to many other variants that have also been studied in the literature. 
Imprecision problems have also been previously studied in one dimension (i.e., the imprecise regions are intervals, as we do in this work).
 For instance, in the 1D $k$-center problem on imprecise points, one is given $n$ intervals on the real line, and the goal is to find $k$ points (centers) on the real line minimizing the maximum distance from the imprecise points to 
 their closest center. This problem can be solved in $O(n)$ time if the imprecise points are sorted according to the midpoints of their intervals~\cite{hu2022computing}.  
The problem of finding the largest minimum color-spanning interval is also related to the 1D \emph{dispersion problem}: given $n$ (uncolored) intervals, choose one point from each interval such that the minimum
distance between any pair of consecutive points is maximized. 
The decision version of the problem (given a value $q$, to decide if there is a solution with no two consecutive points closer than $q$) can be solved in $O(n\log n)$ time~\cite{garey1981scheduling}. 
If the ordering of the selected points along the intervals respects a given ordering of the intervals, the  dispersion problem can be solved in $O(n)$ time~\cite{li2018dispersing}. 
Multiple variants of the dispersion problem exist, depending on the exact objective and constraints taken into account. See, e.g.,~\cite{biedl2021dispersion,fiala2005systems,Neruda} and  references therein.  


As already mentioned, the related work most relevant to ours is that of Acharyya et~al~\cite{acharyya2022minimum}, who studied algorithms for the minimum color-spanning circle problem for imprecise points modeled as disks in 2D. Given $n$ colored disks, they consider finding one point in each disk such that the minimum color-spanning circle of the selected points has a maximum radius.
They show this problem is NP-hard even for unit input disks and only two colors. 
In contrast, the minimization version of the problem, for $k$ colors, can be solved in $O(nk \log n)$ time~\cite{acharyya2022minimum}.  


\subsection{Definitions and notation}

The input to the problem is a set of closed intervals ${\cal I}=\{I_1, I_2,\ldots, I_n\}$  of unit length on the real line, where each interval is colored with one of $k$  colors. 
We will sometimes refer to the intervals also as \emph{segments}.
We assume that no two intervals of the same color are at the same position. 
The intervals are given sorted from left to right, breaking ties arbitrarily if needed. 
For each $i$, we use $x_i$ to denote the left endpoint of $I_i$. We assume that $x_1=0$.

Given a realization (also sometimes called a \emph{representation}) $\cal P$ of  $\cal I$, we call the chosen point of $I_i$ the \emph{representative} of $I_i$. We denote it by $r_i$ (with some abuse of notation, $r_i$ is sometimes used to denote the coordinate of the representative). 
Sometimes, \emph{realization} applies to a proper subset of $\cal I$ rather than the entire $\cal I$. In this case, we say that $I_i$ is \emph{represented} in the realization, if the realization contains a representative of $I_i$.


\section{Disjoint and semi-disjoint case}\label{sec:kda}


\subsection{Disjoint case}

In this subsection, we consider the version where the input segments are pairwise disjoint. In other words, the input is a sequence ${\cal I}=\{I_1, I_2,\ldots, I_n\}$ of $n$ disjoint segments of unit length, sorted with respect to the left endpoint of the segments. 

We use the concept of {\em minimal color-spanning interval} (\mmcsi), defined as a color-spanning interval not properly contained within any other color-spanning interval.
A minimum color-spanning interval is also a minimal color-spanning interval, with minimum length.

In general, \mmcsi s   depend on the  realization.
However, if intervals are  disjoint, the only locations that are important are that of the first and last representatives.
All other intervals represented in the \mmcsi\ are fully contained in the \mmcsi, thus the positions of their representatives are irrelevant.
Hence, from a combinatorial point of view, an \mmcsi\ is a sequence of consecutive intervals.
It follows that the possible combinatorial \mmcsi s are determined by the intervals alone: 


\begin{lemma}
\label{number_ints_disj}
 If the input segments are disjoint, every realization leads to the same set of at most  $n-k+1$ combinatorial \mmcsi s.
\end{lemma}


\begin{proof}
As the input segments in $\cal I$ are disjoint, the ordering of the points in any realization is the same as the left to right order of the segments in $\cal I$. 
Thus, all the realizations lead exactly to the same combinatorial minimal color-spanning intervals.
Indeed, the only thing that depends on the realization is the exact location of the representative of the first and of the last interval of each \mmcsi, because all the intervals in between must be fully contained between the first and last representatives.
As for the number of \mmcsi s, observe that, if we charge each minimal color-spanning interval to the leftmost segment represented, each segment is charged at most once. As there are $n$ segments of $k$ colors and any minimal color-spanning interval contains at least $k$ segments, there are at most $n-(k-1)$ possible leftmost segments.
\end{proof}




The set of combinatorial \mmcsi s can be computed in $O(n)$ time~\cite{chen2013algorithms}. Let $\gamma=\{\gamma_1,\gamma_2,\ldots,\gamma_\ell\}$ be this set, sorted by the leftmost segment represented.
We divide $\gamma$ into subsets called  \emph{chains of \mmcsi s} 
as follows: Start with $\gamma_1$, and let the rightmost segment represented in $\gamma_1$ be $I_\beta$.  If there is an \mmcsi\ $\gamma_i$ starting at $I_\beta$, add it to the subset of $\gamma_1$, and repeat with $\gamma_i$; otherwise, stop. At the end of the process, remove from $\gamma$ all the $\gamma_i$ in the same subset as $\gamma_1$, and construct a new subset starting with the remaining element in $\gamma$ with the smallest index. Repeat this process until $\gamma$ is empty. See Fig.~\ref{fig:dmcsi} for an example.

\begin{figure}[tb]
    \centering
    \includegraphics{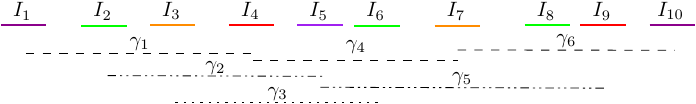}
    \caption{An example where the chains of \mmcsi s are $\{\gamma_1, \gamma_4, \gamma_6\}$, $\{\gamma_2,\gamma_5\}$, $\{\gamma_3\}$.}
     \label{fig:dmcsi}
\end{figure}


The idea of the algorithm is that the problem can be broken into several subproblems (corresponding to the distinct chains of \mmcsi s) because the solutions corresponding to different subproblems do not interfere with each other. In detail, we process each chain of \mmcsi s separately as follows: Let ${\cal I}'$ be the set of segments obtained from taking the leftmost and rightmost segments represented in each $\gamma_i$ contained in the chain of \mmcsi s. For example,  in Fig.~\ref{fig:dmcsi}, for $\{\gamma_1, \gamma_4, \gamma_6\}$ we would have ${\cal I}'=\{I_1,I_4,I_7,I_{10}\}$. Among all \mmcsi s contained in the chain of \mmcsi s, the minimum one is achieved by the pair of representatives of consecutive segments of ${\cal I}'$ at minimum distance. Notice that the position of the representatives of the segments not in ${\cal I}'$ is irrelevant for this subproblem. 
Since we want to maximize the length of the \mcsi s, we compute a realization of ${\cal I}'$ maximizing the minimum distance between consecutive representatives. Such a realization can be computed in time $O(|{\cal I}'|)$~\cite{li2018dispersing}. 

After repeating this procedure for all chains of \mmcsi s, we obtain a realization for all those segments that are the leftmost or rightmost of some $\gamma_i$. The fact that each segment is leftmost (resp., rightmost) for at most one \mmcsi\ implies that each segment receives at most one representative. For all the other segments, we choose a representative arbitrarily. We output the obtained representation as the solution to the problem. We next argue that the algorithm is correct.

\begin{theorem} \label{thm:kda-sol}
If the input segments are pairwise disjoint, the \lmcsi~ problem can be solved in $O(n)$ time.  
\end{theorem}

\begin{proof}
Let ${\cal P^*}$ be a realization solving the \lmcsi\ problem and let $q^*$ be the length of a \mcsi\ of ${\cal P^*}$. For the $i$th chain of \mmcsi s, let $\ell_i$ be the minimum distance between consecutive representatives of the subset ${\cal I}'$ in a realization maximizing such minimum distance. Since any realization of ${\cal I}'$ induces a \mcsi\ of length at most $\ell_i$, we obtain that $q^* \leq \ell_i$. Let $\ell=\min\{\ell_1,\ell_2,\ldots\}$. Then, $q^*\leq \ell$. Since the realization returned by the algorithm has \mcsi\ of length $\ell$, it is a solution to the problem.

Recall that the input segments are sorted. Computing the \mmcsi s and dividing them into chains of \mmcsi s can be done in linear time. Each chain of \mmcsi s can be processed in time linear in its size. Since there are at most $n-k+1$ \mmcsi s, the overall running time of the algorithm is linear. 
\end{proof}

Observe that the presented algorithm does not use the fact that all input segments have the same (unit) length. Therefore, it solves a more general problem:

\begin{corollary} \label{cor:arb-length}
If the input segments have arbitrary lengths and are pairwise disjoint, the \lmcsi~ problem can be solved in $O(n)$ time.  
\end{corollary}

\subsection{Semi-disjoint case}

In this subsection, we show an adaptation of the previous algorithm for the case where intervals might intersect only if they have the same color.


\begin{lemma} \label{lem:semi-dis}
Let $I_i,I_{i+1},\ldots,I_k$ be a continuous subsequence of $\cal I$ with color $c$, and such that $I_{i-1}\cap I_i=\emptyset$, $I_{k}\cap I_{k+1}=\emptyset$, and for all $j\in\{i,i+1,\ldots,k-1\}$ we have $I_{j}\cap I_{j+1}\neq \emptyset$.
\begin{itemize}
\item[(a)] If $I_i\cap I_k=\emptyset$, it is enough to consider realizations where 
$r_i\leq r_j \leq r_k$, for all  $j\in \{i+1,i+2,\ldots,k-1\}$. 
\item[(b)] If $I_i\cap I_k\neq \emptyset$, it is enough to consider realizations where 
$r_i= r_j= r_k$, for all  $j\in \{i+1,i+2,\ldots,k-1\}$. 
\end{itemize}
\end{lemma}

\begin{proof}
Let ${\cal P^*}$ be a realization solving the \lmcsi\ problem. For ease of notation, we use $r_j$ to denote the representative of $I_j$ in ${\cal P^*}$.

Suppose first that $I_i\cap I_k=\emptyset$. Then, $r_i<r_k$. Due to the fact that $I_i, I_{i+1},\ldots, I_k$ is a continuous subsequence of $\cal I$ with color $c$ together with the fact that no two segments of distinct color intersect, we derive that there is no representative of color distinct from $c$ in the interval $[r_i,r_k]$. Since all intervals have the same length, $I_j\cap [r_i,r_k]\neq \emptyset$ for all  $j\in \{i+1,i+2,\ldots,k-1\}$. Thus, $r_j$ can be moved to $I_j\cap [r_i,r_k]$ and this change does not decrease the length of any \mcsi.

Suppose next that $I_i\cap I_k\neq\emptyset$. Then, $\bigcap_{j=i}^{k} I_j \neq \emptyset$. For all $j\in \{i,i+1,\ldots,k\}$, $r_j\in[x_i,x_k+1]$ and the interval $[x_i,x_k+1]$ 
does not contain any representative of color distinct from $c$. Thus, for all $j\in \{i,i+1,\ldots,k\}$, $r_j$ can be moved to a common position in  $\bigcap_{j=i}^{k} I_j$ that also lies between the original positions of $r_i$ and $r_k$, and this change does not decrease the length of any \mcsi.
\end{proof}

We construct a new instance ${\cal I}'$ of the problem as follows: Every segment $I_j\in {\cal I}$ that does not have any intersection with any other segment of $\cal I$ is added to ${\cal I}'$. For a continuous subsequence of segments $I_i,I_{i+1},\ldots,I_k$ of same color and such that $I_{i-1}\cap I_i=\emptyset$, $I_{k}\cap I_{k+1}=\emptyset$, and for all $j\in\{i,i+1,\ldots,k-1\}$ we have $I_{j}\cap I_{j+1}\neq \emptyset$, then: (i) if $I_i\cap I_k=\emptyset$, we (only) add to ${\cal I}'$ the segments $I_i,I_k$; (ii) if $I_i\cap I_k\neq \emptyset$, we add to ${\cal I}'$ a segment $\bigcap_{j=i}^{k} I_j$ of the same color as $I_i$. 

By Lemma~\ref{lem:semi-dis}, to solve the problem for ${\cal I}$ it is enough to solve it for ${\cal I}'$. Additionally, ${\cal I}'$ forms an instance where input segments are pairwise disjoint (but do not necessarily have the same length). By Corollary~\ref{cor:arb-length}, we conclude:

\begin{theorem} \label{thm:sol-semi-dis}
If there is no pair of intersecting intervals of distinct color, the \lmcsi~ problem can be solved in $O(n)$ time.  
\end{theorem}



\section{Preliminaries for the case $k=2$}\label{sec:two-color_preli}

In the rest of the paper, we focus on the case where $k=2$.


By Theorems~\ref{thm:kda-sol} and~\ref{thm:sol-semi-dis}, if the segments are pairwise disjoint or there is no pair of intersecting intervals of distinct color, the problem can be solved in $O(n)$ time.  
Therefore, in the rest of the paper we restrict to the case where there is a pair of intersecting intervals of distinct color.

We assume the two colors are red and blue, denoting the set of red and blue intervals by ${\cal R}$ and  ${\cal B}$, respectively. Let $q^*$ be the length of an \mcsi~in a realization solving the \lmcsi~problem. 

\begin{lemma}
If at least two intervals of different colors intersect, then $\frac{1}{2}\leq q^* \leq 2$.
\end{lemma}
\begin{proof}
To prove that $q^*\geq \frac{1}{2}$, observe that every red interval contains at least one integer point and every blue interval contains at least one point of the form $m+\frac{1}{2}$, where $m\in \mathbb{Z}$. We pick such points as representatives. This choice of realization ensures that any color-spanning interval has a length of at least $\frac{1}{2}$.

The bound $q^*\leq 2$ follows trivially from the fact that there is a pair of intersecting intervals of distinct color. The two representatives of this pair form a color-spanning interval, and this interval has a length of at most 2.
\end{proof}

We observe that these bounds are best possible: The upper bound is achieved when the input consists of two intervals of distinct color that intersect at exactly one point. Regarding the lower bound, we take a blue interval with left endpoint at $x=0$ and a collection of red intervals centered at $x=0,\, x=\varepsilon,\, x=2\varepsilon,\ldots,\,x=1$, for a small value $\varepsilon>0$. In any representation, the blue representative has a red center at distance $\leq \varepsilon/2$, so the corresponding red representative is at distance $\leq \frac{1}{2}+\frac{\varepsilon}{2}$.


Roughly speaking, our algorithm consists of an algorithm for the decision problem (given some value $q$, do we have $q^*\geq q$ or $q^*< q$?) combined with a binary search on a set of candidates for $q^*$. We divide the range $\left(\frac{1}{2},2\right]$ into three subranges and present a different algorithm for each of them. 

The decision version of the problem can be rephrased as: For a given value $q$, does there exist a realization of $\cal R \cup \cal B$ such that the distance between any pair of representatives of distinct color is at least $q$?
Such a realization is said to satisfy the \emph{separation property}. 

Suppose that, in the sorted sequence of input intervals, two consecutive intervals leave a gap of length $q$ or greater between them. Then we can divide the problem into two independent subproblems, one containing the intervals to the left of the gap and the other containing the intervals to the right. Indeed, the answer to the original problem is yes if and only if the answer to both subproblems is yes. Hence, to solve the decision problem for a fixed value of $q$, we can assume that no two consecutive intervals leave a gap of length $q$ or greater between them. This assumption is made in the subsequent sections dealing with the decision problem.

The main difficulty of the problem is that, if two intervals of distinct color overlap by $q$ or more, there are two possible orderings of the representatives that yield a distance between them of at least $q$. This is also the reason why our algorithms for the decision problem become more complex as $q$ decreases.

All of our algorithms for the decision problem use the following concept:

\begin{definition}
A \emph{leftmost solution} is a solution where each representative is either at the left endpoint of its interval, or at a distance $q$ from a representative of distinct color located to its left.
\end{definition}

Given any solution, we can transform it into a leftmost solution by traversing the representatives from left to right and moving them as much as possible to the left, while maintaining a representation satisfying the separation property. This implies the following:

\begin{observation} \label{obs:left-sol-enough}
If there is a solution, there is one that is leftmost.
\end{observation}

Our algorithms compute a leftmost solution, provided it exists.

In the following sections, we use $c$ or $c'$ to denote one of the two colors. For a fixed $c$, we use $\bar{c}$ to denote the color that is not $c$. We use $c_i$ to denote the color of~$I_i$.

 \section{Decision problem for $q\in \left(\frac{3}{4},2\right]$}

\subsection{Decision problem for $q\in \left(1,2\right]$}

This is an easy case due to the following observation:

\begin{lemma}
In the decision problem for $q\in \left(1,2\right]$, it is enough to consider realizations where $r_i\leq r_{i+1}$, for all $i=1,2,\ldots,n-1$.
\end{lemma}

\begin{proof}
Suppose that the decision problem has an affirmative answer, and let the realization $\{r_1,r_2,$ $\ldots,r_n\}$ be a certificate.
If $I_i\cap I_{i+1} = \emptyset$, clearly $r_i<r_{i+1}$. If $I_i\cap I_{i+1} \neq \emptyset$, $I_i,I_{i+1}$ have the same color and $r_i>r_{i+1}$, we can reassign the representatives ($r_{i+1}$ becomes the representative of $I_i$ and $r_{i}$ becomes the representative of $I_{i+1}$), and this change does not
decrease the length of any \mcsi. If $I_i\cap I_{i+1} \neq \emptyset$ and $I_i,I_{i+1}$ have different color, the fact that $r_i,r_{i+1}$ are at distance at least $q$, with $q>1$, implies that it is not possible that $r_i>r_{i+1}$.
\end{proof}

Observe that, if ${\cal R}$ contains two intervals of distinct color at the same position, the answer to the decision problem is negative. Thus, we can assume that this situation does not occur. Additionally, one of the assumptions of our problem is that there are no two intervals of the same color at the same position. Hence, the left-to-right order of the input segments is a total order. We process the segments in this order and try to compute a leftmost solution. The algorithm stops if, for some segment $I_i$, there is no placement for $r_i$ in $I_i$ satisfying the separation property (with respect to the representatives that have already been placed). It is easy to see that the algorithm succeeds in finding a representation if and only if the answer to the decision problem is affirmative.

\begin{proposition} \label{prop:dec-easy}
For $k=2$, the decision problem for $q\in \left(1,2\right]$ can be solved in $O(n)$ time.
\end{proposition}

\subsection{Decision problem for $q\in \left(\frac 3 4, 1\right]$}

In this case, there might be more left-to-right orderings for the representatives leading to affirmative answers for the problem. To deal with them, we introduce a bipartite graph $H$ with vertices $\cal R \cup \cal B$ where a red and blue segment are adjacent if and only if the length of their intersection is at least~$q$. We denote the connected components of $H$ by $C_1,C_2\ldots$ and their sets of vertices by $V(C_1),V(C_2)\ldots$

\begin{lemma}\label{lem:max-length_span}
Suppose that we sort the segments of ${\cal R}\cap V(C_i)$ (resp, ${\cal B}\cap V(C_i)$) from left to right, and we take two consecutive segments $I_j$ and $I_k$. Then, $x_k-x_j\leq 2-2q$. 
\end{lemma}

\begin{proof}
Assume that $I_j$ and $I_k$ are red. Since they are consecutive, ${\cal R}\cap V(C_i)={\cal R}_1 \cup {\cal R}_2$, where ${\cal R}_1=\{I_m\in {\cal R}\cap V(C_i) \textrm{ s.t. } x_m\leq x_j\}$ and ${\cal R}_2=\{I_m\in {\cal R}\cap V(C_i) \textrm{ s.t. } x_m\geq x_k\}$. A path from $I_j$ to $I_k$ in $C_i$ starts at ${\cal R}_1$ and finishes at ${\cal R}_2$. Let us take a chain $I_{j'},I_b,I_{k'}$ of consecutive vertices in the path such that $I_{j'}\in {\cal R}_1$, $I_{b}\in {\cal B}\cap V(C_i)$ and $I_{k'}\in {\cal R}_2$. 
If $x_b\leq x_j$, since $I_{k'}$ and $I_b$ are adjacent in $H$, we have that $x_b\geq x_{k'}-(1-q)$. Since $x_j\geq x_b$ and $x_{k'}\geq x_k$, we obtain that $x_k-x_j\leq 1-q\leq 2-2q$. Analogously, if $x_b\geq x_k$, we obtain the same conclusion.
If $x_j<x_b<x_k$, the facts that the pairs $I_{j'},I_b$ and $I_{k'},I_b$ are adjacent in $H$ imply that $x_b\leq x_{j'}+1-q\leq x_{j}+1-q$ and $x_b\geq x_{k'}-(1-q)\geq x_{k}-(1-q)$. Thus, $x_k-x_j\leq 2-2q$.
\end{proof}

We define ${\cal J}_{C_i}$ as the interval spanned by $V(C_i)$, that is, its leftmost (resp., rightmost) point is the leftmost (resp., rightmost) point of the leftmost (resp., rightmost) segment(s) in $V(C_i)$.

\begin{lemma} \label{lem:graphH}
For any $I_j\in \cal R \cup \cal B$, if $I_j\subseteq {\cal J}_{C_i}$, then $I_j\in V(C_i)$.
\end{lemma}

\begin{proof}
Suppose, for the sake of contradiction, that $I_j\notin V(C_i)$. Let $I_j$ be red.  

We start by arguing that there exists some $I_k\in {\cal B}\cap V(C_i)$ such that $x_k>x_j$. Suppose that this is not the case. Let $I_b$ be the rightmost segment in ${\cal B}\cap V(C_i)$. Then, $x_b<x_j$\footnote{$x_b=x_j$ is not possible because then both intervals would overlap by more than $q$ and $I_j\in V(C_i)$}. Since $I_j\notin V(C_i)$, $|I_j\cap I_b|<q$ and $x_j>x_b+1-q$. Let $I_r$ be the rightmost segment in ${\cal R}\cap V(C_i)$. Then, $x_r\leq x_b+1-q$ because otherwise $I_r$ would not intersect any segment in ${\cal B}\cap V(C_i)$ by at least $q$. Hence, all segments in $V(C_i)$ start to the left of $x_j$, contradicting the fact that $I_j\subseteq {\cal J}_{C_i}$.

Among all segments in ${\cal B}\cap V(C_i)$ starting to the right of $x_j$, let $I_k$ be the leftmost one. Since $I_j\notin V(C_i)$, $x_k>x_j+1-q$. By analogous arguments, there exist segments in ${\cal B}\cap V(C_i)$ starting to the left of $x_j$, and the rightmost one, called $I_{k'}$, satisfies $x_{k'}<x_j -(1-q)$. But then $I_{k'}$ and $I_k$ are consecutive segments in ${\cal B}\cap V(C_i)$ with $x_k-x_{k'}>2-2q$, contradicting Lemma~\ref{lem:max-length_span}.
\end{proof}



\begin{lemma}
For any two intervals ${\cal J}_{C_i}$ and ${\cal J}_{C_j}$ with $i\neq j$, their left endpoints do not coincide.
\end{lemma}

\begin{proof}
With some abuse of notation, we denote by $I_i$ (resp, $I_j$) an interval of $V(C_i)$ (resp. $V(C_j)$) whose leftmost endpoint coincides with the leftmost point of ${\cal J}_{C_i}$ (resp., ${\cal J}_{C_j}$). Then $x_i\neq x_j$: Otherwise, $I_i$ and $I_j$ would have opposite colors (recall that the input does not contain two segments of the same color at the same position) and would be at the same position, so they would overlap by at least $q$ and they would be in the same connected component of $H$.
\end{proof}

In consequence, we can uniquely sort the intervals ${\cal J}_{C_i}$ from left to right according to their left endpoints. We rename the components $C_i$ so that this left-to-right order is ${\cal J}_{C_1},{\cal J}_{C_2}\ldots$

\begin{lemma} \label{lem:key-int-q}
In the decision problem for $q\in \left(\frac 3 4, 1\right]$, it is enough to consider realizations where, for any $I_{i'}\in C_i$ and $I_{j'}\in C_j$ with $i<j$, $r_{i'}\leq r_{j'}$. 
\end{lemma}




\begin{proof}
Let the realization $\{r_1,r_2,\ldots,r_n\}$ be a yes-certificate for the problem.

If ${\cal J}_{C_i}\cap {\cal J}_{C_j}=\emptyset$, the conclusion is clear. Otherwise, we denote the positions of the left and right endpoints of ${\cal J}_{C_i}$ by $\ell({\cal J}_{C_i})$ and $r({\cal J}_{C_i})$, respectively. We start by showing that $x_{i'}\leq x_{j'}$. Indeed, if $x_{i'}> x_{j'}$, we would have $x_{j'}\geq \ell({\cal J}_{C_j})> \ell({\cal J}_{C_i})$ and $x_{j'}<x_{i'}\leq r({\cal J}_{C_i})-1$. Thus,  $I_{j'}\subseteq {\cal J}_{C_i}$, contradicting Lemma~\ref{lem:graphH}. 

If $c_{i'}\neq c_{j'}$, $r_{i'}\leq r_{j'}$ is implied by $x_{i'}\leq x_{j'}$ and $|I_{i'} \cap I_{j'}|<q$. 
If $c_{i'}=c_{j'}$ and $r_{i'} > r_{j'}$, the fact that $x_{i'}\leq x_{j'}$ allows us to reassign the representatives ($r_{i'}$ becomes the representative of $I_{j'}$ and vice versa), and this change does not decrease the length of \mcsi(s).
\end{proof}



The previous lemma suggests a greedy approach similar to the one in the previous subsection. Before presenting it, it remains to explain how to find representations of $V(C_i)$ satisfying the separation property.

\begin{lemma} \label{lem:alt-middle-q}
In any realization of $V(C_i)$ satisfying the separation property, there is exactly one alternation between red and blue representatives.
\end{lemma}

\begin{proof}
We start by proving this auxiliary claim: In any realization of $V(C_i)$ satisfying the separation property, for any red interval $I_r$  in $V(C_i)$, the representatives of the blue intervals in $V(C_i)$ adjacent to $I_r$ in $H$ are all on the same side of $r_r$.\footnote{Obviously, the claim also holds exchanging red and blue.} Indeed, let $I_b,I_{b'}$ be two blue intervals in $V(C_i)$ adjacent to $I_r$ in $H$. Suppose, for the sake of contradiction, that there exists a realization of $V(C_i)$ satisfying the separation property where $r_b<r_r<r_{b'}$. Since 
the pairs $I_b,I_r$ and $I_{b'},I_r$ are adjacent in $H$, we have $x_b\geq x_r-(1-q)$ and $x_{b'}\leq x_r+(1-q)$. Hence, $r_{b'}-r_b\leq x_r+2-q-x_r+1-q=3-2q<2q$, since $q>3/4$. Thus, the triple $r_b,r_r,r_{b'}$ does not satisfy the separation property.

We next move to the statement of the lemma. Assume, for the sake of contradiction, that in some realization satisfying the separation property we have some triple $r_r<r_b<r_{r'}$, where $I_r,I_{r'}\in {\cal R}\cap V(C_i)$ and $I_b\in {\cal B}\cap V(C_i)$. As we have seen, the representatives of the red intervals in $V(C_i)$ adjacent to $I_b$ in $H$ are all on the same side of $r_b$, say to its left. This means that $I_{r'}$ is not a neighbor of $I_b$ in $H$. It also implies that all blue representatives have the representatives of their neighbors to their left, and all red representatives have the representatives of their neighbors to their right.

Let us take a path $P$ in $C_i$ between $I_{r'}$ and $I_b$, and let us traverse it starting from $I_{r'}$. We define $I_{b'}$ as the first blue interval in $P$ such that $r_{b'}<r_{r'}$ (it might happen that $I_{b'}=I_b$). Notice that $I_{b'}$ is not a neighbor of $I_{r'}$ in $H$. Hence, any path $P'$ from $I_{r'}$ to $I_{b'}$ in $H$ consists of at least three edges. See Fig.~\ref{fig:lem_9}.

\begin{figure}
    \centering
    \includegraphics{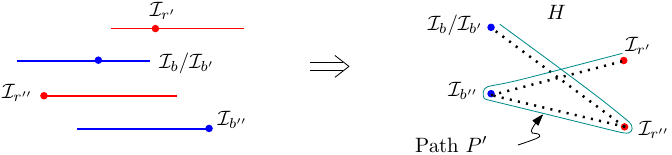}
    \caption{Illustration for the proof of Lemma~\ref{lem:alt-middle-q}.}
    \label{fig:lem_9}
\end{figure}

Denote by $I_{r''}$ the neighbor of $I_{b'}$ in $P'$. As we have argued earlier, $r_{r''}< r_{b'}$. By hypothesis, $r_{b'}<r_{r'}$. Let $I_{b''}$ be the neighbor of $I_{r''}$ in $P'$ different from $I_{b'}$. Since all the blue segments that are in the interior of $P'$ have representatives to the right of $r_{r'}$, we obtain that $r_{r'}<r_{b''}$. Hence, we have $r_{r''}<r_{b'}<r_{r'}<r_{b''}$. Since the realization satisfies the separation property, we derive that $r_{b''}-r_{r''}\geq 3q$. However, $I_{b''}$ and $I_{r''}$ are adjacent in $H$, so $r_{b''}-r_{r''}\leq 2-q$.
Since $q>3/4$, we reach a contradiction.
\end{proof}

Recall that we assume that $x_1=0$. Our algorithm uses the subroutine \texttt{subrep}, given in Algorithm~\ref{alg:subroutine}. Roughly speaking, this subroutine is used to find representations of $V(C_i)$ that can be concatenated with some fixed representation of $V(C_{i-1})$ and where the color of the rightmost representative is fixed. In detail, the subroutine takes as input $(i,c,c',x)$, and returns a leftmost representation of $V(C_i)$ where the color of the rightmost representative is $c$,  where the separation property is satisfied within the representation and also with respect to a representative of color $c'$ placed at $x$ (corresponding to the rightmost representative in some precomputed representation of $V(C_{i-1})$), and where all representatives are at a position greater than or equal to $x$ (recall Lemma~\ref{lem:key-int-q}).

Since the color of the rightmost representative is required to be $c$, by Lemma~\ref{lem:alt-middle-q} in the representation we first encounter all representatives of color $\bar{c}$ and afterwards all representatives of color $c$. For this reason, the subroutine processes first all intervals of color $\bar{c}$ and afterwards all intervals of color $c$.

\begin{algorithm}
\begin{algorithmic}[1]
\vspace{0.25cm}
\Require  $(i,c,c',x)$
\Ensure  a leftmost representation of $V(C_i)$ where the color of the rightmost representative is $c$, where the separation property is satisfied within the representation and also with respect to a representative of color $c'$ placed at $x$, and where all representatives are at a position greater or equal than $x$ (if such a representation exists)
\If{$x=+\infty$}
\State stop and return $null$
\EndIf
\If{$c'=c$} \Comment{In this case, intervals of color $\bar{c}$ must be at distance at least $q$ from $x$}
\For{all intervals $I_j\in V(C_i)$ of color $\bar{c}$}
\State $r_j:=\max\{x+q,x_j\}$
\If{$r_j>x_j+1$}
\State stop and return $null$
\EndIf
\EndFor
\State $y:=$ maximum among all the $r_j$'s found in the previous loop
\For{all intervals $I_j\in V(C_i)$ of color $c$}
\State $r_j:=\max\{y+q,x_j\}$
\If{$r_j>x_j+1$}
\State stop and return $null$
\EndIf
\EndFor
\Else \Comment{In this case, intervals of color $\bar{c}$ must be on $x$ or to its right}
\For{all intervals $I_j\in V(C_i)$ of color $\bar{c}$}
\State $r_j:=\max\{x,x_j\}$
\If{$r_j>x_j+1$}
\State stop and return $null$
\EndIf
\EndFor
\State $y:=$ maximum among all the $r_j$'s found in the previous loop
\For{all intervals $I_j\in V(C_i)$ of color $c$}
\State $r_j:=\max\{y+q,x_j\}$
\If{$r_j>x_j+1$}
\State stop and return $null$
\EndIf
\EndFor
\EndIf
\State return the values $r_j$'s
\end{algorithmic}
\caption{Subroutine \texttt{subrep}}
\label{alg:subroutine}
\end{algorithm}

The main algorithm is outlined in Algorithm~\ref{alg:alg-middle-q}. Given a representation of some $V(C_i)$, in the pseudocode we use $f(\cdot)$ to denote the position of the rightmost representative in the representation. If the representation is $null$, we use the convention $f(null)=+\infty$.

The algorithm processes the components of $H$ in order $C_1,C_2\ldots$. For every $C_i$, using the subroutine \texttt{subrep}, it computes (if they exist) two leftmost representations of $V(C_i)$ satisfying the separation property: one where the rightmost representative is red (called $P^r[i]$ in the pseudocode), and the other where it is blue  (called $P^b[i]$). For each of $P^r[i]$ or $P^b[i]$, there are two candidates: either the preceding representation of $V(C_{i-1})$ finishes with a red representative (representation $P_1$ in the pseudocode), or with a blue representative (representation $P_2$). Among both, the algorithm chooses one minimizing $f(\cdot)$. It also keeps variables $prev^r[i]$ and $prev^r[i]$ that store which of $P_1$ or $P_2$ was chosen (or, more precisely, if the associated representation of $V(C_{i-1})$ finishes with a red or blue representative), and that allow to reconstruct the final solution. Finally, the variable $lastr$ (resp., $lastb$) stores the position of the rightmost representative in $P^r[i]$ (resp., $P^b[i]$).

\begin{algorithm}
\begin{algorithmic}[1]
\vspace{0.25cm}
\Require  $\cal I$ 
\Ensure  a solution, if it exists; otherwise, ``there is no solution"
 \State compute $V(C_1),V(C_2),\ldots, V(C_m)$ \label{line:ini1-middle-q} \Comment{$C_1,\ldots,C_m$ are sorted according to the left-to-right order of the left endpoint of ${\cal J}_{C_i}$}
 \State $lastr:=-1$, $lastb:=-1$ \Comment{$lastr$ (resp., $lastb$) store the position of the rightmost representative in the computed representation of $V(C_{i-1})$ finishing with a red (resp., blue) representative}
\For{$i=1,2,\ldots,m$}
\State $P_1:=\texttt{subrep}(i,r,r,lastr)$
\State $P_2:=\texttt{subrep}(i,r,b,lastb)$
\State $P^r[i]:=$ the one between $P_1,P_2$ minimizing $f(\cdot)$
\State $prev^r[i]:=$ $r,$ if $P_1$ minimizes $f(\cdot)$; $b,$ if $P_2$ minimizes $f(\cdot)$; $undef$, if $f(P_1)=+\infty$, $f(P_2)=+\infty$
\State $P_1:=\texttt{subrep}(i,b,r,lastr)$ \label{line:alg2-line}
\State $P_2:=\texttt{subrep}(i,b,b,lastb)$
\State $P^b[i]:=$ the one between $P_1,P_2$ minimizing $f(\cdot)$
\State $prev^b[i]:=$ $r,$ if $P_1$ minimizes $f(\cdot)$; $b,$ if $P_2$ minimizes $f(\cdot)$; $undef$, if $f(P_1)=+\infty$, $f(P_2)=+\infty$
\State $lastr:=f(P^r[i])$, $lastb:=f(P^b[i])$
\If{$lastr=+\infty$ and $lastb=+\infty$}
\State stop and return ``there is no solution"
\EndIf
\EndFor
\State reconstruct a solution
\end{algorithmic}
\caption{Decision problem for $q\in \left(\frac 3 4, 1\right]$}
\label{alg:alg-middle-q}
\end{algorithm}

\begin{theorem} \label{thm:dec-mid}
For $k=2$, the decision problem for $q\in \left(\frac 3 4, 1\right]$ can be solved in $O(n)$ time.
\end{theorem}

\begin{proof}
We first argue that Algorithm~\ref{alg:alg-middle-q} is correct. Clearly, if the algorithm returns a representation, such representation is a solution to the problem. 

Conversely, by Observation~\ref{obs:left-sol-enough}, Lemma~\ref{lem:key-int-q} and Lemma~\ref{lem:alt-middle-q}, if there is a solution, there is a leftmost solution $LS$ with the shape considered by the algorithm. We prove the following statement by induction: Suppose that, in $LS$, the representation of $V(C_i)$ finishes with a blue representative. Then the algorithm computes some $P^b[i]\neq null$, and in $P^b[i]$ every representative is at the same position of the corresponding representative in $LS$ or to its left. An analogous statement holds for red color.

For $i=1$, the representation of $V(C_1)$ in $LS$ is equal to \texttt{subrep}$(1,b,r,-1)$ and \texttt{subrep}$(1,b,b,-1)$, so it is equal to $P^b[1]$. The reason for this is that the left-to-right order of the representatives is the same in both representations, and the leftmost representation for a prescribed ordering of the representatives is unique. Next, suppose that, in $LS$, the representations of $V(C_{i-1})$ and $V(C_i)$ finish with, say, a red and a blue representative, respectively. By induction hypothesis, the algorithm computes some $P^r[i-1]\neq null$, and in $P^r[i-1]$ every representative is at the same position of the corresponding representative in $LS$ or to its left. In line~\ref{line:alg2-line}, $P_1$ is set to \texttt{subrep}$(i,b,r,lastr)$, where $lastr$ is the position of the rightmost representative of $P^r[i-1]$. Since this representative is at the same position of the corresponding representative in $LS$ or to its left, the position of every representative in \texttt{subrep}$(i,b,r,lastr)$ is also set to the same position of the corresponding representative in $LS$ or to its left. Since the algorithm chooses as $P^b[i]$ the one between \texttt{subrep}$(i,b,r,lastr)$ and \texttt{subrep}$(i,b,b,lastr)$ minimizing the position of the rightmost representative, the conclusion follows.

Regarding the running time, we start by describing the computation of the connected components of $H$: We process the input segments in left-to-right order. Let the color of $I_i$ be $c$. Any segment $I_k$ with $x_i-(1-q)\leq x_k \leq x_i+(1-q)$ and with color $\bar{c}$ is in the same component as $I_i$. We make a linear scan in the set of segments of color $\bar{c}$ starting from the leftmost unprocessed segment and add to the component the ones satisfying the condition above. We repeat this process with the rightmost segment of the component and alternating colors, $c$ and $\bar{c}$, at each repetition. If at some repetition no new segment is added to the component, we start a new component from the interval of smallest index that has not been added to any component yet. We repeat this process until all input segments get processed. This generates all the connected components of $H$ in $O(n)$ time. Additionally, a call \texttt{subrep}$(i,\cdot,\cdot,\cdot)$ takes $O(|V(C_i)|)$. Since the algorithm performs at most four calls of this type for every $i$, the total running time is $O(n)$.
\end{proof}

 \section{Decision problem for $q\in \left(\frac{1}{2},\frac 3 4\right]$}
 \label{sec:dec-long}

This case requires a more involved algorithm and demands the use of a more sophisticated algorithmic machinery to achieve an efficient running time.

\subsection{Foundations of the algorithm}

We derive some additional properties of leftmost solutions.

\begin{definition}
A position $x$ is \emph{forbidden} for color $c$ if there exists an interval of color $\bar{c}$ with leftmost point in the interval $(x-q,x+q-1)$. A position $x$ is \emph{forbidden} for the representative of interval $I_i$ if $x$ is forbidden for $c_i$.
\end{definition}

Notice that, if $x$ is a forbidden position for the representative of $I_i$, there exists an interval $I_j$ of the opposite color such that $x$ divides $I_j$ into two portions of length smaller than $q$. Thus, choosing $x$ as a representative for $I_i$ makes it impossible to find a representative for $I_j$ not violating the separation property. We conclude that, while aiming for a solution, no representative can be placed at a forbidden position.

\begin{definition}
A representation of a subset of intervals is \emph{valid} if it satisfies the separation property and no representative is placed at a forbidden position. 
\end{definition}

We point out that, even though the previous definition can be applied to a proper subset of the set of intervals, the notion of \emph{forbidden} position is always defined with respect to the whole set of intervals.

Our first goal is to show that leftmost solutions are composed of pieces of the following shape:

\begin{definition} \label{def:tab-sol}
A \emph{tabular subsolution ${\cal{T}}_i$ starting at $I_i$} is a valid representation of a subset of intervals including $I_i$ such that: (i) $r_i=x_i$; (ii) for all representatives $r_j$ of ${\cal{T}}_i$ with $j\neq i$, we have $r_j>r_i$ and $r_j$ is placed at a distance $q$ to the right of a representative of opposite color; (iii) if for two intervals of the same color $I_k, I_j$ we have that ${\cal{T}}_i$ contains a representative $r_k$ for $I_k$, $k\neq i$ and $r_k \in I_j$, then ${\cal{T}}_i$ contains a representative $r_j$ for $I_j$ and $r_j=r_k$.
\end{definition}

A tabular subsolution starting at $I_i$ will sometimes be called a tabular subsolution \emph{with index~$i$}. 

Suppose that $I_i$ is red. Notice that there does not exist any tabular subsolution starting at $I_i$ if $x_i$ is a forbidden position for the representative of $I_i$. If this is not the case, a tabular subsolution starting at $I_i$ is formed by the red representative $r_i$ placed at $x_i$, possibly followed by a set of blue representatives placed at $x_i+q$, which is in turn possibly followed by a set of red representatives placed at $x_i+2q$, and so on. An example is given in Fig.~\ref{fig:tabular_solution}.

\begin{figure} 
    \centering
    \includegraphics{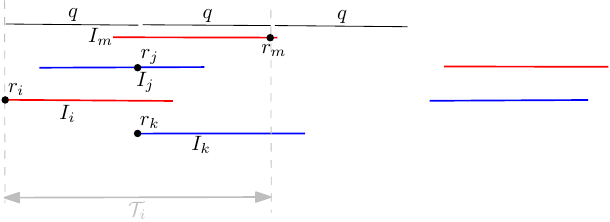}
    \caption{A tabular solution ${\cal{T}}_i$. We note ${\cal{T}}_i$ is maximal and illustrates case (i)  in Observation~\ref{obs:end-Mi}.}
    \label{fig:tabular_solution}
\end{figure}


We can now give a characterization of leftmost solutions:

\begin{lemma}  \label{lem:prun-tab}
Any leftmost solution is a sequence of non-overlaping tabular subsolutions.
\end{lemma} 

\begin{proof}
We select the indices of every interval $I_i$ such that in the leftmost solution $r_i=x_i$ and $r_i$ \emph{is not}
at a distance $q$ to the right of a representative of opposite color. Let $i_1<i_2<\cdots<i_m$ be the set of selected indices. Notice that the left-to-right order of the representatives of $I_{i_1},I_{i_2},\ldots,I_{i_m}$ is the same as the left-to-right order of the intervals. We want to prove that, for every $j\in \{1,2,\ldots,m-1\}$, the representatives contained in the interval $[x_{i_j},x_{i_{j+1}})$ form a tabular subsolution with index~$i_j$. Using the same arguments, one can prove that the representatives contained in the interval $[x_{i_m},+\infty)$ also form a tabular subsolution.

Let $I^{+}_{i_j}$ be the set of intervals whose representative is contained in the interval $[x_{i_j},x_{i_{j+1}})$. Since the given representation is a solution, the representation of $I^{+}_{i_j}$ in $[x_{i_j},x_{i_{j+1}})$ is valid. Additionally, by definition, $r_{i_j}=x_{i_j}$, i.e., condition (i) of Definition~\ref{def:tab-sol} holds. 

We next show that there does not exist any $I_k\in I^{+}_{i_j}\setminus I_{i_j}$ such that $r_k=r_{i_j}$. If $c_k\neq c_{i_j}$, the conclusion is clear. If $c_k=c_{i_j}$, we suppose for the sake of contradiction that $r_k=r_{i_j}$. Since there are no two intervals of the same color at the same position, we derive that $r_k\neq x_k$. Since the solution is leftmost, $r_k$ is at a distance $q$ from a representative of opposite color located to its left. Thus, so is $r_{i_j}$, which contradicts the way in which we have selected $i_j$. In consequence, for all $I_k\in I^{+}_{i_j}\setminus I_{i_j}$ we have that $r_k>r_{i_j}$. By construction of the selected indices and definition of leftmost solution, for all intervals in $I^{+}_{i_j}\setminus I_{i_j}$ their representatives are placed at a distance $q$ from a representative of distinct color located to its left. Thus, condition (ii) of Definition~\ref{def:tab-sol} also holds. 

It only remains to verify condition (iii) of Definition~\ref{def:tab-sol}. Let $I_k\in I^{+}_{i_j}\setminus I_{i_j}$ (if it exists). Since there is a representative of color $\bar{c_k}$ at $r_k-q$, the next point to the left of $r_k$ (if any) where there are representatives of color $c_k$ (in the given solution) lies in the interval $(-\infty,r_k-2q]$. Since $2q>1$, any interval $I_{\ell}$ of color $c_k$ containing $r_k$ does not contain any point in $(-\infty,r_k-2q]$. Thus, $r_{\ell}\geq r_k$. Suppose, for the sake of contradiction, that $r_{\ell}> r_k$. Since $I_{\ell}$ contains $r_k$ and $r_{\ell}> r_k$, we derive that $r_{\ell}\neq x_{\ell}$. Thus, $r_{\ell}$ is at distance $q$ from a representative of color $\bar{c_k}$ located to its left. In consequence, $r_{\ell}\in[r_k+2q,\infty)$. Again, since $I_{\ell}$ contains $r_k$, it does not contain any point in $[r_k+2q,\infty)$. Hence, in the given leftmost solution, $I_{\ell}$ has its representative at $r_k$ as required in condition (iii).
\end{proof}


The previous lemma justifies the strategy of our algorithm: Try to construct a sequence of tabular subsolutions such that the separation property is satisfied and the sequence contains a representative for every interval. An example is given in Figure~\ref{fig:concatenating}. Before presenting our algorithm, we need to show additional properties of tabular subsolutions.

\begin{figure}[tb]
    \begin{subfigure}{\linewidth}
    \centering
    \includegraphics{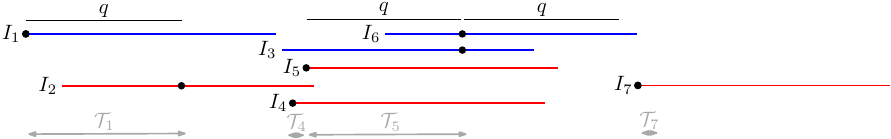}
       \caption{A leftmost solution composed of four tabular subsolutions ${\cal{T}}_1,{\cal{T}}_4, {\cal{T}}_5,{\cal{T}}_7$.}
    \label{fig:concatenating-1}
    \end{subfigure}
    \\[1cc]
    \begin{subfigure}{\linewidth}
    \centering
    \includegraphics{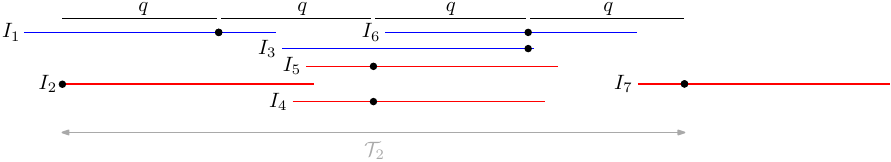}
       \caption{Another leftmost solution, composed of only one tabular subsolution, ${\cal{T}}_2$.}
    \label{fig:concatenating-2}
    \end{subfigure}
    \caption{Two leftmost solutions for the same input.}
    \label{fig:concatenating}
\end{figure}

\begin{lemma} \label{lem:pres-ord-tab}
\begin{itemize}
\item [(a)] Let $I_j,I_k$, with $j<k$, have the same color. If $I_j$ and $I_k$ have representative in some tabular subsolution ${\cal{T}}_i$, $r_j\leq r_k$. 
\item [(b)] The intervals of the same color represented in a tabular subsolution ${\cal{T}}_i$ form a continuous sequence (i.e., no intermediate interval is skipped).
\end{itemize}
\end{lemma}

\begin{proof}
(a) If $I_j\cap I_k=\emptyset$, the claim is trivial. Otherwise, if $j=i$ the claim follows from condition (ii) of Definition~\ref{def:tab-sol}. Otherwise assume, for the sake of contradiction, that $r_j>r_k$. This, together with $j<k$, implies that $r_k$ is contained in $I_j$. But then condition (iii) of Definition~\ref{def:tab-sol} implies that $r_k=r_j$.

(b) The proofs for the two colors are almost identical, so we prove the statement for $c_i$. By definition, the intervals of color $c_i$ represented in ${\cal{T}}_i$ are $I_i$ (which can be seen as the rightmost interval in $\cal I$ containing $x_i$), all intervals of color $c_i$ containing $x_i+2q$ (which form a continuous sequence),\ldots, and all intervals of color $c_i$ containing $x_i+2\nu q$ (which form a continuous sequence), for some positive integer $\nu$. Hence, if the statement were not true, there would exist some $I_k$ of color $c_i$ and some $\rho\in\{0,1,\ldots,\nu -1\}$ such that $I_k$ is to the right of $x_i+2\rho q$ and to the left of $x_i+2(\rho+1)q$. By definition of ${\cal{T}}_i$, there is a representative of color $\bar{c_i}$ at position $x_i+(2\rho +1) q$. However, this position is forbidden due to $x_i+2(\rho+1)q-1>x_k>x_i+2\rho q$. We obtain a contradiction.
%
%
\end{proof}

We point out that the complete set of intervals whose representative appears in ${\cal{T}}_i$ does not necessarily form a set of consecutive intervals. For example, in Figure~\ref{fig:concatenating-1}, ${\cal{T}}_5$ contains representatives for $I_3,I_5,I_6$.

Notice that we have defined tabular subsolutions in a way that there might be many tabular subsolutions starting at $I_i$. To be able to work with them in an efficient way, we introduce the concept of \emph{maximal} tabular subsolution. Intuitively, a maximal tabular subsolution is a tabular solution that cannot be extended further to the right. The formal definition is as follows:

\begin{definition} The \emph{maximal tabular subsolution starting at $I_i$}, denoted $M[i]$, is the tabular subsolution starting at $I_i$ containing the largest number of representatives of intervals. \end{definition}



We can observe the following.

\begin{observation} \label{obs:end-Mi}
Suppose that the rightmost representatives of $M[i]$ have color $c$ and are placed at $x_i+kq$, for some integer $k\geq 0$. Then:
(i) there is no interval of color $\bar{c}$ that contains the point $x_i+(k+1)q$; or
(ii) there is an interval of color $\bar{c}$ that contains the point $x_i+(k+1)q$, but the position $x_i+(k+1)q$ is forbidden for color $\bar{c}$.
\end{observation}

For example, in Figure~\ref{fig:tabular_solution}, $M[i]={\cal{T}}_i$ and it cannot be extended further due to case (i). In Figure~\ref{fig:concatenating-1}, $M[4]={\cal{T}}_4$ contains only one point due to case (ii): extending it further requires placing the representatives of $I_3$ and $I_6$ at $x_4+q$, which is a forbidden position (it would leave no feasible representative for $I_5$).

It can happen that $M[i]=\emptyset$. It can also happen that  $M[i]$ includes representatives at $x_i+q$ of intervals that begin to the left of $x_i$ (for example, in Figure~\ref{fig:concatenating-2},  $M[2]={\cal{T}}_2$ contains a representative for $I_1$).

We use $\max^r[i]$ and $\max^b[i]$, respectively, to denote the index of the rightmost red and rightmost blue interval with representative in $M[i]$, respectively.
If $M[i]$ does not contain a representative for any interval of some color, we set the parameter to $-1$.




Next, we classify the set of intervals with the same color as $I_i$, as follows: 
\begin{itemize}
    \item We denote by $\mathcal{I}^{s,<}_i$ the set of intervals $I_j$ such that $x_j<x_i$ (i.e., $j<i$).
    \item We denote by $\mathcal{I}^{s,>}_i$ the set of intervals $I_j$ such that $x_j>x_i$ (i.e., $j>i$).
\end{itemize}

\begin{figure}
    \centering   \includegraphics[scale=1]{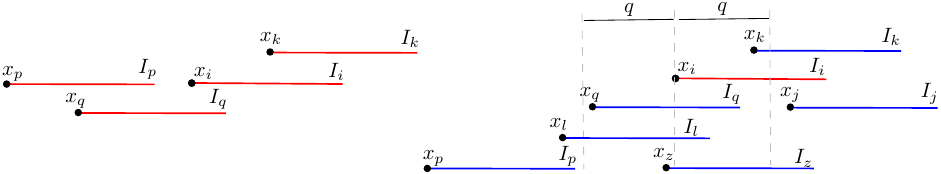}
    \caption{Illustration of the classification of intervals with respect to $I_i$. Left: $\mathcal{I}^{s,<}_i=\{I_p,I_q\}$, 
    $\mathcal{I}^{s,>}_i=\{I_k\}$. Right: 
    $\mathcal{I}^{o,<<}_i=\{I_p\}$, 
      $\mathcal{I}^{o,<}_i=\{I_l\}$, 
   $\mathcal{I}^{o,f}_i=\{I_q\}$, 
     $\mathcal{I}^{o,>}_i=\{I_z\}$, 
  $\mathcal{I}^{o,>>}_i=\{I_k,I_j\}$.}
    \label{fig:all_intervals}
\end{figure}

We also classify the intervals with color opposite from that of $I_i$, as follows: 
\begin{itemize}
    \item We denote by  $\mathcal{I}^{o,<<}_i$ the set of intervals $I_j$ such that $x_j\leq x_i-(1+q)$.
    \item We denote by  $\mathcal{I}^{o,<}_i$ the set of intervals $I_j$ such that $x_j \in(x_i-(1+q)),x_i-q]$.
    \item We denote by $\mathcal{I}^{o,f}_i$ the set of intervals $I_j$ such that $x_j\in (x_i-q,x_i-(1-q))$.
    \item We denote by $\mathcal{I}^{o,>}_i$ the set of intervals $I_j$ such that $x_j\in [x_i-(1-q),x_i)$.
    \item We denote by $\mathcal{I}^{o,>>}_i$ the set of intervals $I_j$ such that $x_j\in (x_i,+\infty)$.
\end{itemize}

See Fig.~\ref{fig:all_intervals}  for an illustration. 
     If there exists $I_j$ with color opposite from that of $I_i$ such that $x_j=x_i$, then we place it in $\mathcal{I}^{o,>}_i$ if $j<i$, and we place it in $\mathcal{I}^{o,>>}_i$ if $j>i$.\footnote{Recall that pairs of intervals of different color at the same position are sorted arbitrarily.}

For each of these families, we use the notation $\mathcal{\overrightarrow{\mathcal{I}}}^{\cdot,\cdot}_i$ and $\mathcal{\overleftarrow{\mathcal{I}}}^{\cdot,\cdot}_i$ to denote the rightmost and leftmost interval of that family, respectively. Given some interval $I_j$, we denote its index by $\mathrm{ind}(I_j)$ (i.e., $\mathrm{ind}(I_j)=j$).  If the argument of $\mathrm{ind}(\cdot)$ is the rightmost or leftmost interval of an empty set, we set the value of $\mathrm{ind}(\cdot)$ to $-1$.

We want to understand under which conditions $M[i]$ (or part of it) might belong to a leftmost solution. We start by observing that, if $\mathcal{I}^{o,f}_i\neq \emptyset$, then $M[i]=\emptyset$ because $x_i$ is a forbidden position. 

\begin{lemma} \label{lem:conditions}
Let $LS$ be a leftmost solution. Let us decompose $LS$ into a sequence of non-overlaping tabular subsolutions as in the proof of Lemma~\ref{lem:prun-tab}, that is, a new tabular subsolution begins at every interval $I_i$ such that in $LS$ $r_i=x_i$ and $r_i$ \emph{is not}
at a distance $q$ to the right of a representative of opposite color. Let ${\cal{T}}_l$ be the leftmost tabular subsolution of $LS$, and ${\cal{T}}_i$ be any other tabular subsolution of LS.
Denote by $L[i]$ the set of representatives in $LS$ to the left of $x_i$. The following conditions are satisfied:
\begin{itemize} 
    \item[(C0)] $\mathcal{I}^{s,<}_l=\mathcal{I}^{o,<<}_l=\mathcal{I}^{o,<}_l=\emptyset$.
    \item[(C1)] For all $I_j\in \mathcal{I}^{s,<}_i$, $r_j\in L[i]$. 
    \item[(C2)] For all $I_j\in \mathcal{I}^{o,<<}_i$, $r_j\in L[i]$. 
    \item[(C3)] For all $I_j\in \mathcal{I}^{o,<}_i$, $r_j\in L[i]$ and $r_j\leq x_i-q$. 
\end{itemize}
\end{lemma}

\begin{proof}
(C0): Intervals from the set $ \mathcal{I}^{o,<<}_l$ lie completely to the left of $x_l$, so any solution places its representatives to the left of $x_l$. Since $LS$ is a solution and the leftmost representative of $LS$ is placed at $x_l$, we conclude that $\mathcal{I}^{o,<<}_l=\emptyset$. 

Next suppose, for the sake of contradiction, that $\mathcal{I}^{s,<}_l\neq \emptyset$ and let $I_j\in\mathcal{I}^{s,<}_l$. If $x_j<x_i-1$ (i.e., $I_j$ lies completely to the left of $I_i$), the previous argument applies. Otherwise, $x_j\in[x_i-1,x_i)$ (i.e., $I_j$ contains $x_i$). Since the leftmost representative of $LS$ is placed at $x_l$, $r_j$ is placed at $x_l$ or to its right. Thus, $r_j\neq x_j$ and $r_j$ is placed at a distance $q$ from a representative of distinct color located to its left. Since the leftmost representative of $LS$ is placed at $x_l$, such representative of distinct color is to the right of $x_l$. Hence, $r_j$ is at distance at least $2q$ from $x_l$, which is not possible because $I_j$ contains both points and has length 1 ($<2q$).

Finally, regarding $\mathcal{I}^{o,<}_l$, these intervals have a portion of length smaller than $q$ on or to the right of $x_l$ and have color opposite to that of $I_l$. Thus, any solution that sets $r_l=x_l$ places the representatives of intervals in $\mathcal{I}^{o,<}_l$ to the left of $x_l$. Since $LS$ starts at $x_l$, we again conclude that $\mathcal{I}^{o,<}_l=\emptyset$.

The proofs of (C1), (C2) and (C3) are similar to the analogous cases in condition (C0).

(C1): Let $I_j\in \mathcal{I}^{s,<}_i$.  If $x_j<x_i-1$, we have that $r_j<r_i=x_i$. Since $M[i]$ only contains representatives on or to the right of $x_i$, $r_j$ is necessarily contained in $L[i]$. If $x_j\in[x_i-1,x_i)$, there are two possibilities: (i) If $r_j=x_j$, we automatically conclude that $r_j\in L[i]$. (ii) If $r_j$ is placed at a distance $q$ from a representative of distinct color $r_k$ located to its left, as in the analogous case in (C0), $r_k$ cannot be to the right of $x_i$. Hence $r_k<x_i=r_i$. Since $LS$ is a solution, $r_k\leq r_i-q$. If $r_k<r_i-q$, then $r_j<r_i$ and we are done. If $r_k=r_i-q$, then $r_i$ is at a distance $q$ from a representative of distinct color $r_k$ located to its left, which contradicts the way in which we have subdivided the leftmost solution into tabular subsolutions.

(C2): Intervals in $\mathcal{I}^{o,<<}_i$ are to the left of $I_i$ and leave a gap with $I_i$ greater than or equal to $q$.  Therefore, no matter where the representative is, it is compatible with the representative of $I_i$ at $x_i$. The representative is to the left of $x_i$, so it is necessarily contained in $L[i]$.

(C3): Let $I_j\in\mathcal{I}^{o,<<}_i$. This interval have a portion of length smaller than $q$ on or to the right of $x_i$. Thus, if $r_i=x_i$, $r_j$ can only be placed to the left of $x_i$ because $I_i$ and $I_j$ have distinct colors. Additionally, the separation property enforces $r_j\leq x_i-q$.
\end{proof}

\begin{definition}
Let $k<i$. A tabular subsolution ${\cal{T}}_i$ is \emph{compatible} with a tabular subsolution ${\cal{T}}_k$ if the rightmost representative of ${\cal{T}}_k$ is to the left of $x_i$ and every pair of representatives of different colors, one of which is in ${\cal{T}}_i$ and the other in ${\cal{T}}_k$, are at distance at least $q$.
\end{definition}

\begin{figure}
    \centering
    \includegraphics{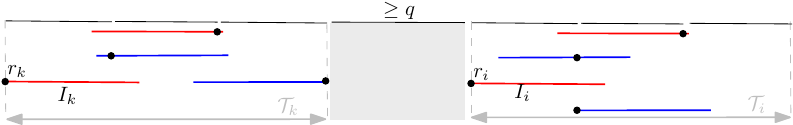}
    \caption{${\cal{T}}_k$ and ${\cal{T}}_i$ are compatible.}
    \label{fig:compatible_tabular_solution}
\end{figure}

Fig.~\ref{fig:compatible_tabular_solution} shows an example of two compatible tabular subsolutions.

Conditions (C0)-(C3) have been defined for a leftmost solution, but they can be defined analogously for a given sequence of tabular subsolutions. In the following, we also discuss the satisfiability of these conditions for sequences of tabular subsolutions for which we do not necessarily know that they form a solution.

\begin{definition}
A sequence ${\cal{S}}$ of tabular subsolutions is \emph{super valid} if it satisfies the following: (i) the leftmost tabular subsolution 
of ${\cal{S}}$ satisfies condition (C0); 
(ii) for every pair of consecutive tabular subsolutions ${\cal{T}}_k,{\cal{T}}_i$, with $k<i$, ${\cal{T}}_k$ and ${\cal{T}}_i$ are compatible, and ${\cal{T}}_i$ satisfies conditions (C1), (C2) and (C3).
\end{definition}

We have seen in Lemma~\ref{lem:conditions} that (C0)-(C3) are necessary conditions for a leftmost solution. In general terms, the next goal is to prove that they are also sufficient. Before, we need another technical lemma:

\begin{lemma} \label{lem:seq-cons-2}
Let ${\cal{S}}$ be a super valid sequence of tabular subsolutions. 
Suppose that the rightmost representative of ${\cal{S}}$ is at position $x$ and has color $c$. Then, 
${\cal{S}}$ contains representatives exactly for all intervals $I_j$ of color $c$ such that $x_j\leq x$ and for all intervals $I_b$ of color $\bar{c}$ such that $x_b< x+q-1$.
\end{lemma}

\begin{proof}
Let us start with color $c$. Suppose, for the sake of contradiction, that there exists an interval $I_j$ of color $c$ such that $x_j\leq x$ and ${\cal{S}}$ does not contain any representative for $I_j$. We distinguish two main cases.

Suppose first that there exists some $r_k\in {\cal{S}}$ such that $c_k=c$ and $I_j$ contains $r_k$. Then, condition (iii) of Definition~\ref{def:tab-sol} together with the fact that ${\cal{S}}$ does not contain any representative for $I_j$ imply that ${\cal{S}}$ contains a tabular subsolution ${\cal{T}}_k$ of index $k$, and thus $r_k=x_k$. 
Since $I_j$ contains $r_k$($=x_k$), $I_j\in \mathcal{I}^{s,<}_k$. If ${\cal{T}}_k$ is the leftmost tabular subsolution of ${\cal{S}}$, this contradicts the fact that ${\cal{T}}_k$ satisfies condition (C0). Otherwise, since $L[k]$ does not contain any representative for $I_j$, it contradicts the fact that ${\cal{T}}_k$ satisfies condition (C1).

Suppose next that $I_j$ does not contain any representative of ${\cal{S}}$ of color $c$. Let $r_k$ be a leftmost representative of ${\cal{S}}$ of color $c$ to the right of $I_j$ (notice that $r_k$ exists because in this case $x$ is a representative of ${\cal{S}}$ of color $c$ to the right of $I_j$). If ${\cal{S}}$ contains a tabular subsolution of index $k$, we apply the same arguments as in the paragraph above and reach a contradiction. Otherwise, $r_k$ is at distance $q$ from a representative $r_m$ of color $\bar{c}$ located to its left. We observe that ${\cal{S}}$ contains a tabular subsolution ${\cal{T}}_m$: Otherwise, there would be a representative $r_p$ of color $c$ at $r_m-q$, but $r_p>x_j+1$ is not possible (by definition of $r_k$), $x_j\leq r_p\leq x_j+1$ is also not possible (because $I_j$ does not contain any representative of color $c$), and $r_p<x_j$ is also not possible (because then $r_m$ would be a forbidden position for color $\bar{c}$ due to $I_j$).

Since $x_j+1<r_k$ and $r_m=r_k-q$, we derive that $r_m>x_j-q+1$. Since all points in the interval $(x_j-q+1,x_j+q)$ are forbidden positions for $\bar{c}$, we obtain that $x_m=r_m\geq x_j+q$. Thus, $I_j\in \mathcal{I}^{o,<<}_m\cup \mathcal{I}^{o,<}_m$. If ${\cal{T}}_m$ is the leftmost tabular subsolution of ${\cal{S}}$, this contradicts the fact that ${\cal{T}}_m$ satisfies condition (C0). Otherwise, since $L[m]$ does not contain any representative for $I_j$, it either contradicts the fact that ${\cal{T}}_m$ satisfies (C2) or that it satisfies (C3).

It remains to argue about color $\bar{c}$. Suppose that there exists an interval $I_b$ of color $\bar{c}$ such that $x_b< x+q-1$ and ${\cal{S}}$ does not contain any representative for $I_b$. If there exists some $r_k\in {\cal{S}}$ such that $c_k=\bar{c}$ and $I_b$ contains $r_k$, or if there exists a representative of ${\cal{S}}$ of color $\bar{c}$ to the right of $I_b$, we apply the same arguments as will color $c$.

In the remaining case for color $\bar{c}$, observe that $x_b< x+q-1$ and the fact that $x$ is not a forbidden position for color $c$ imply that $x_b\leq x-q$. If $x$ is the leftmost representative of some tabular subsolution ${\cal{T}}_m$, this implies that $I_b\in \mathcal{I}^{o,<<}_m\cup \mathcal{I}^{o,<}_m$, contradicting one of the necessary conditions. Otherwise, ${\cal{S}}$ contains a representative of color $\bar{c}$ at $x-q$. Since, in our current case, such representative cannot be to the right of or contained in $I_b$, we have $x-q<x_b$, which contradicts the fact that $x_b\leq x-q$.
\end{proof}

The following is a direct consequence of the previous lemma.

\begin{corollary}\label{cor:seq-cons-2}
Let ${\cal{S}}$ be a super valid sequence of tabular subsolutions. 
    For each color, ${\cal{S}}$ 
     contains representatives for a consecutive set of intervals of that color, including the leftmost interval of the color.
\end{corollary}


\begin{lemma} \label{lem:alm-sol}
Let ${\cal{S}}$ be a super valid sequence of tabular subsolutions such that the rightmost tabular subsolution is a maximal tabular subsolution containing a representative for $I_n$. Then, $\cal{S}$ is a solution to the decision problem.
\end{lemma}

\begin{proof}
For $\cal{S}$ to be a solution, we need to prove that it satisfies the separation property and that it contains representatives for all intervals.

Within a tabular subsolution the separation property is obviously satisfied. Additionally, the fact that consecutive subsolutions are compatible guarantees that the separation property is satisfied between pairs of consecutive subsolutions. This is enough to guarantee that the property is satisfied in the entire sequence ${\cal{S}}$.

Let us next prove that $\cal{S}$ contains representatives for all intervals. Let ${\cal{T}}_i$ be the rightmost tabular subsolution of ${\cal{S}}$, and $x$ be the position of the representative for $I_n$ in ${\cal{T}}_i$. Suppose that $I_n$ has color $c$. 
By Corollary~\ref{cor:seq-cons-2}, ${\cal{S}}$ contains representatives for a consecutive set of intervals of color $c$ including $I_n$ and the leftmost interval of color $c$. Thus, ${\cal{S}}$ contains representatives for all intervals of color $c$. To prove that it also contains representatives for all intervals of color $\bar{c}$, we consider two cases. In both cases, we denote by $I_{n'}$ the rightmost interval of color $\bar{c}$.

Suppose first that $x$ is the position of the rightmost representative of $\cal{S}$. By Lemma~\ref{lem:seq-cons-2}, ${\cal{S}}$ contains representatives for all intervals $I_b$ of color $\bar{c}$ such that $x_b< x+q-1$. Thus, if $x_{n'}<x+q-1$, we are done. Notice that $x_{n'}>x+q$ is not possible, because we would have $x_n\leq x<x+q<x_{n'}$, which contradicts the fact that $I_n$ is a rightmost interval from the set. Hence, the remaining case is $x_{n'}\in [x+q-1,x+q]$. Then $I_{n'}$ contains $x+q$. Since $x$ is the rightmost position of ${\cal{T}}_i$ and ${\cal{T}}_i$ is maximal, ${\cal{T}}_i$ does not contain representatives at $x+q$ because this is a forbidden position for color $\bar{c}$. In consequence, there exists an interval of color $c$ with leftmost point in $(x,x+2q-1)$. This contradicts the fact that $I_n$ is a rightmost interval from the set.

Finally, suppose that $x$ is not the position of the rightmost representative of $\cal{S}$. 
Notice that the rightmost representative of $\cal{S}$ cannot be at $x+2q,x+3q\ldots$, because this would contradict the fact that $I_n$ is a rightmost interval from the set. Thus, the rightmost representative of $\cal{S}$ is at $x+q$. Since all intervals $I_b$ of color $\bar{c}$ satisfy $x_b\leq x_n\leq x$ and at least one of them contains $x+q$, we derive that $I_{n'}$ contains $x+q$. Thus, $\cal{S}$ contains a representative for $I_{n'}$ and, by Corollary~\ref{cor:seq-cons-2}, ${\cal{S}}$  contains representatives for all intervals of color $\bar{c}$.
\end{proof}

Lemma~\ref{lem:alm-sol} is the key lemma used by our algorithm to construct a solution. However, checking (C1)-(C3) every time that we want to concatenate tabular subsolutions is an expensive operation. Our final step before presenting the algorithm is to show that, when $\mathcal{I}^{o,<}_i\neq \emptyset$, checking (C3) is enough. Before, we need to introduce a few definitions.

\begin{definition}
Given a sequence ${\cal{S}}$ of tabular subsolutions, an interval $I_j$ is \emph{redundant} for ${\cal{S}}$ if there exists a tabular subsolution ${\cal{T}}_i$ of ${\cal{S}}$ different from the leftmost one such that $I_j\in \mathcal{I}^{s,<}_i$ and $I_j$ contains $x_i$.
\end{definition}

For example, in Figure~\ref{fig:concatenating-1}, $I_4$ is redundant for the sequence ${\cal{S}}=\{{\cal{T}}_1,{\cal{T}}_5,{\cal{T}}_7\}$.

\begin{definition}
A sequence ${\cal{S}}$ of tabular subsolutions is \emph{almost valid} if it satisfies the same conditions of a super valid sequence except for the fact that some intervals that are redundant for ${\cal{S}}$ do not have a representative in ${\cal{S}}$.
\end{definition}


Notice that several types of conditions for several indices might be violated in an almost valid sequence. A redundant interval without representative makes (C1) be violated for the relevant $i$ and possibly for other tabular subsolutions with index greater than $i$. Additionally, it might also make (C2) and (C3) be violated for other tabular subsolutions.

Notice also that, for the purpose of finding a solution, almost valid sequences are also sufficient: The separation property is satisfied, and the sequence contains representatives for all intervals except for some redundant ones. Let $I_j$ be redundant for ${\cal{S}}$ and such that $I_j\in \mathcal{I}^{s,<}_i$, $I_j$ contains $x_i$ and ${\cal{S}}$ contains a tabular subsolution of index $i$ different from the leftmost one. Since $I_j$ has the same color as $I_i$ and contains $x_i$, we can simply add a representative for $I_j$ at $x_i$. If we are interested in a leftmost solution, afterwards we move $r_j$ as much to the left as possible (and we obtain that $r_j=x_j$ or $r_j$ is at a distance $q$ from a representative of distinct color located to its left).

\begin{definition}
A sequence ${\cal{S}}$ of tabular subsolutions is \emph{valid} if it is super valid or almost valid.
\end{definition}


Our final key lemma deals with the case where we need to test if some $M[i]\neq \emptyset$ can be concatenated with a given valid sequence ${\cal{S}}$ of tabular subsolutions, and we have that condition (C3) must be checked because $\mathcal{I}^{o,<}_i\neq \emptyset$. Lemma~\ref{lem:C4-enough} states that, in this case, checking (C3) for $\mathcal{\overrightarrow{\mathcal{I}}}^{o,<}_i$ is enough (in particular, we do not need to check (C1) and (C2)).

\begin{definition}
Let $i,k$ be such that $k<i$, $M[k]\neq \emptyset$, and $M[i]\neq \emptyset$. The \emph{longest compatible} tabular subsolution of $M[k]$ with respect to $M[i]$ is the longest tabular subsolution starting at $I_k$ that is compatible with $M[i]$.
\end{definition}


\begin{lemma} \label{lem:C4-enough}
Let ${\cal{S}}$ be a valid sequence of tabular subsolutions. Let $M[k]$ be the rightmost subsolution of ${\cal{S}}$, $I_i$ be an interval of color $c$ with $i>k$, and ${\cal{T}}_k$ be the longest compatible tabular subsolution of $M[k]$ with respect to $M[i]$. Suppose that $M[i]\neq \emptyset$, ${\cal{T}}_k\neq \emptyset$, and $\mathcal{I}^{o,<}_i\neq \emptyset$. We define ${\cal{S'}}= ({\cal{S}} \setminus M[k]) \cup {\cal{T}}_k \cup M[i]$. 
If ${\cal{S'}}$ contains a representative for $\mathcal{\overrightarrow{\mathcal{I}}}^{o,<}_i$  at a position $\leq x_i-q$, then ${\cal{S'}}$ forms a valid sequence of tabular subsolutions.
\end{lemma}

\begin{proof}
To prove the claim, we need to show that, in ${\cal{S'}}$, $M[i]$ satisfies conditions (C1)-(C3) except for redundant intervals for ${\cal{S'}}$. We define $I_m=\mathcal{\overrightarrow{\mathcal{I}}}^{o,<}_i$.

Conditions (C2) and (C3) are easy: Let $I_s$ be an interval in $\mathcal{I}^{o,<<}_i\cup \mathcal{I}^{o,<}_i$ that is not redundant for ${\cal{S'}}$. Notice that $s<m$. Since $I_m$ has a representative in ${\cal{S'}}$, Corollary~\ref{cor:seq-cons-2} together with the fact that $s<m$ imply that $I_s$ also has a representative in ${\cal{S'}}$. If $I_s\in \mathcal{I}^{o,<}_i$, we additionally need to prove that $r_s\leq x_i-q$. If $r_s$ is contained in ${\cal{T}}_k$, by Lemma~\ref{lem:pres-ord-tab}a we have that $r_s\leq r_m\leq x_i-q$. If it is contained in another tabular subsolution of $\mathcal{S'}$, then it lies to the left of $r_m$, so also at a position $\leq x_i-q$.


To prove that condition (C1) is satisfied, we distinguish several cases.

Let $I_j\in \mathcal{I}^{s,<}_i$ be an interval that is not redundant for ${\cal{S'}}$. We know that $I_m$ has a representative in ${\cal{S'}}$, and we want to show that $I_j$ also does. Notice that the portion of  ${\cal{S'}}$ containing all representatives on or to the left of $r_m$ is also a valid sequence of tabular subsolutions. By Lemma~\ref{lem:seq-cons-2}, such a portion contains representatives for all intervals $I_b$ of color $c$ such that $x_b< r_m+q-1$ and $I_b$ is not redundant for ${\cal{S'}}$. Hence, if $x_j< r_m+q-1$, such a portion contains a representative for $I_j$, and thus the whole ${\cal{S'}}$ also does.

It remains to look at the case where $x_j\geq r_m+q-1$. We first show that $x_j\leq r_m+q$: Since $I_m\in \mathcal{I}^{o,<}_i$, we have that $r_m\geq x_m>x_i-1-q$. If $x_j> r_m+q$, we would obtain that $x_j>x_i-1$. We obtain a contradiction because this implies that $I_j$ contains $x_i$ (so $I_j$ is redundant for ${\cal{S'}}$). Thus, $r_m+q-1\leq x_j\leq r_m+q$. This implies that $I_j$ contains $r_m+q$. Recall that $r_m\leq x_i-q$. If $r_m=x_i-q$, $I_j$ contains $x_i$ and we again obtain a contradiction. Hence, in the remaining of the proof we can assume that $r_m< x_i-q$.

If $M[k]$ has representatives of color $c$ at position $r_m+q$, such representatives are included in ${\cal{T}}_k$ because their position is compatible with a representative of color $c$ at $x_i$. Hence, in this case, $I_j$ has a representative in ${\cal{S'}}$. The only remaining case is when $M[k]$ does not have representatives of color $c$ at position $r_m+q$. We will show that such a case leads to a contradiction. Since there is at least one interval of color $c$ containing $r_m+q$ (interval $I_j$), we derive that $M[k]$ does not have representatives there because $r_m+q$ is a forbidden position. Thus, there exists an interval $I_h$ of color $\bar{c}$ such that $x_h\in (r_m, r_m+2q-1)$. Since $r_m< x_i-q$, we have that $x_h<x_i+q-1$. Since $M[i]\neq \emptyset$, $\mathcal{I}^{o,f}_i=\emptyset$ and thus $x_h\leq x_i-q$. Additionally, $x_h>r_m\geq x_m$. Thus, $I_h\in \mathcal{I}^{o,<}_i$. Since $x_h>x_m$, we obtain a contradiction with the fact that $I_m= \mathcal{\overrightarrow{\mathcal{I}}}^{o,<}_i$.
\end{proof}

\subsection{Algorithm}

\begin{definition}
We say that a non-empty maximal tabular subsolution $M[i]$ is \emph{appendable} if there exists a valid sequence of tabular subsolutions with $M[i]$ as the rightmost tabular subsolution.
\end{definition}

The algorithm is given in Algorithm~\ref{alg:alg-main}. For each $i$ with $M[i]\neq \emptyset$, the algorithm determines if $M[i]$ is appendable. To this end, there are two main options: If $\mathcal{I}^{s,<}_i=\mathcal{I}^{o,<<}_i=\mathcal{I}^{o,<}_i=\emptyset$ (condition (C0)), $M[i]$ is appendable because on its own it forms a valid sequence of tabular subsolutions (lines~\ref{line:leftmost-1}-\ref{line:leftmost-2}). Otherwise, the algorithm looks for some $k<i$ such that $M[k]$ is appendable, and some portion ${\cal{T}}_k$ of $M[k]$ is compatible with $M[i]$ and is such that conditions (C1)-(C3) are satisfied (and then it sets $p[i]:=k)$. Here there are two cases. If $\mathcal{I}^{o,<}_i\neq \emptyset$, to satisfy (C1)-(C3) we simply check that $M[k]$ contains a representative for $\mathcal{\overrightarrow{\mathcal{I}}}^{o,<}_i$ at a position $\leq x_i-q$ (lines~\ref{line:notempty-1}-\ref{line:last-1}). Otherwise, we check that the valid sequence of tabular subsolutions finishing with $M[k]$ (witnessing that $M[k]$ is appendable) contains representatives for $\mathcal{\overrightarrow{\mathcal{I}}}^{s,<}_i$ and $\mathcal{\overrightarrow{\mathcal{I}}}^{o,<<}_i$ (lines~\ref{line:last-case}-\ref{line:last-2}). Finally, the problem has a solution if and only if the algorithm finds some $i$ such that $M[i]$ is appendable and contains a representative for $I_n$ (lines~\ref{line:imp}-\ref{line:final}).

In the next lemmas we will use the following notation: We denote by $\mathcal{I}^{s,<<}_i$ the set of intervals $I_j$ such that $I_j\in \mathcal{I}^{s,<}_i$ and $I_j$ does not contain $x_i$.

\begin{lemma} \label{lem:corr-lemma}
If the algorithm sets $\mathrm{append}[i]=T$, then $M[i]\neq \emptyset$ and there exists a valid sequence of tabular subsolutions finishing with $M[i]$. Additionally, the indices of such a valid sequence can be obtained by reversing the sequence $i,p[i],p[p[i]]\ldots$ and possibly removing some of the indices. 
\end{lemma}

\begin{proof} 
We prove the claim by induction on $i$. We notice that the algorithm only sets $\mathrm{append}[i]=T$ when $M[i]\neq \emptyset$, so we focus on the rest of the claim.

If $i=1$, we have that $\mathcal{I}^{s,<}_i=\mathcal{I}^{o,<<}_i=\mathcal{I}^{o,<}_i=\emptyset$. In this case, the algorithm sets $\mathrm{append}[i]=T$. By definition, the sequence formed only by $M[1]$ forms a valid sequence of tabular subsolutions.
If $i>1$ and $\mathcal{I}^{s,<}_i=\mathcal{I}^{o,<<}_i=\mathcal{I}^{o,<}_i=\emptyset$, the same argument applies.

If $i>1$ and $\mathcal{I}^{s,<}_i=\mathcal{I}^{o,<<}_i=\mathcal{I}^{o,<}_i=\emptyset$ does not hold, there are two main cases. Suppose first that $\mathcal{I}^{o,<}_i\neq \emptyset$. The algorithm only sets $\mathrm{append}[i]=T$ if it finds some $k\in\{1,2,\ldots,i-1\}$ such that $\mathrm{append}[k]=T$ and $M[k]$ contains a representative for $\mathcal{\overrightarrow{\mathcal{I}}}^{o,<}_i$ at a position $\leq x_i-q$. 
By the induction hypothesis, there exists a valid sequence of tabular subsolutions $\mathcal{S}$ finishing with $M[k]$. Let ${\cal{T}}_k$ be the longest compatible tabular subsolution of $M[k]$ with respect to $M[i]$. Notice that ${\cal{T}}_k\neq \emptyset$ because it contains, at least, the representative for $\mathcal{\overrightarrow{\mathcal{I}}}^{o,<}_i$ (since the position of this representative is to the left of $x_i$ and compatible with $M[i]$). Let us define  ${\cal{S'}}= ({\cal{S}} \setminus M[k]) \cup {\cal{T}}_k \cup M[i]$. Since, in ${\cal{S'}}$, $\mathcal{\overrightarrow{\mathcal{I}}}^{o,<}_i$ contains a representative at a position $\leq x_i-q$, by Lemma~\ref{lem:C4-enough} we obtain that ${\cal{S'}}$ forms a valid sequence of tabular subsolutions. The indices of this valid sequence can be obtained by taking the sequence of indices for $\mathcal{S}$ and adding $i$ at the end.

Finally, suppose that $i>1$, $\mathcal{I}^{s,<}_i=\mathcal{I}^{o,<<}_i=\mathcal{I}^{o,<}_i=\emptyset$ does not hold, and $\mathcal{I}^{o,<}_i= \emptyset$. Without loss of generality, let us suppose that $I_i$ is red. In this case, the algorithm only sets $\mathrm{append}[i]=T$ if it finds some $k\in\{1,2,\ldots,i-1\}$ such that $\mathrm{append}[k]=T$ and the valid sequence of tabular subsolutions $\mathcal{S}$ finishing with $M[k]$ (which exists by induction hypothesis) satisfies the following: the index of the rightmost red and blue interval with representative in $\mathcal{S}$ is greater than or equal to $\mathrm{ind}(\mathcal{\overrightarrow{\mathcal{I}}}^{s,<}_i)$ and $\mathrm{ind}(\mathcal{\overrightarrow{\mathcal{I}}}^{o,<<}_i)$, respectively. The difference with the previous case is that now we do not know whether the relevant representatives are in $M[k]$ or in previous tabular subsolutions of $\mathcal{S}$, so we produce a new sequence $\mathcal{\tilde{S}}$ containing only the relevant portion of $\mathcal{S}$ to be concatenated with $M[i]$. Initially, $\mathcal{\tilde{S}}=\mathcal{S}$ and then we possibly reduce $\mathcal{\tilde{S}}$ as follows:
Let ${\cal{T}}_k$ be the longest compatible tabular subsolution of $M[k]$ with respect to $M[i]$. If ${\cal{T}}_k\neq \emptyset$, replace $M[k]$ by ${\cal{T}}_k$ in $\mathcal{\tilde{S}}$, and stop. Otherwise, remove $M[k]$ from $\mathcal{\tilde{S}}$, and repeat with the new rightmost tabular subsolution of the obtained sequence. 
At the end of the proof, we will discuss the case when, at the end of this procedure, $\mathcal{\tilde{S}}$ is empty. For now, let us assume it is not, let $k'$ be the index of  the rightmost tabular subsolution of $\mathcal{\tilde{S}}$ at the end of this procedure, and let ${\cal{T}}_{k'}$ be the longest compatible tabular subsolution of $M[k']$ with respect to $M[i]$. Notice that $\mathcal{\tilde{S}}$ continues to be valid. Additionally, ${\cal{T}}_{k'}$ and $M[i]$ are compatible. Thus, it only remains to show that, in $\mathcal{\tilde{S}}\cup M[i]$, $M[i]$ satisfies conditions (C1), (C2) and (C3) except for intervals that are redundant for $\mathcal{\tilde{S}}\cup M[i]$. Since $\mathcal{I}^{o,<}_i= \emptyset$, we do not need to worry about (C3).

Let us start with (C1). Since all intervals $I_j\in \mathcal{I}^{s,<}_i$ such that $I_j$ contains $x_i$ are redundant for $\mathcal{\tilde{S}}\cup M[i]$, it is enough to check (C1) for intervals in $\mathcal{I}^{s,<<}_i$. If this set is empty, we are done. Otherwise, since the index of the rightmost red interval with representative in $\mathcal{S}$ is greater than or equal to $\mathrm{ind}(\mathcal{\overrightarrow{\mathcal{I}}}^{s,<}_i)$, by Corollary~\ref{cor:seq-cons-2}, $\mathcal{S}$ contains a representative for $\mathcal{\overrightarrow{\mathcal{I}}}^{s,<<}_i$ unless this interval is redundant for $\mathcal{S}$. Suppose first that it is not. Then $\mathcal{S}$ contains a representative for it. Notice that so does $\mathcal{\tilde{S}}\cup M[i]$ because the position of the representative is compatible with $M[i]$. Thus, by Corollary~\ref{cor:seq-cons-2}, in $\mathcal{\tilde{S}}\cup M[i]$, $M[i]$ satisfies (C1) except for intervals that are redundant for the sequence. Suppose next that $\mathcal{\overrightarrow{\mathcal{I}}}^{s,<<}_i$ is redundant for $\mathcal{S}$. Then there exists some $k'$ such that $\mathcal{\overrightarrow{\mathcal{I}}}^{s,<<}_i\in \mathcal{I}^{s,<}_{k'}$, $\mathcal{\overrightarrow{\mathcal{I}}}^{s,<<}_i$ contains $x_{k'}$ and $\mathcal{S}$ contains a tabular subsolution with index $k'$. Since $k<i$, we also have $k'\leq k <i$. Thus, $\mathcal{S}$ contains a representative for $I_{k'}$ at $x_{k'}<x_i$. Again, so does $\mathcal{\tilde{S}}\cup M[i]$ because the position of the representative is compatible with $M[i]$. Since intervals in $\mathcal{I}^{s,<<}_i$ have index smaller than $k'$, by Corollary~\ref{cor:seq-cons-2}, in $\mathcal{\tilde{S}}\cup M[i]$, $M[i]$ satisfies (C1), except for intervals that are redundant for the sequence.

We next argue about (C2). Since the index of the rightmost blue interval with representative in $\mathcal{S}$ is greater than or equal to $\mathrm{ind}(\mathcal{\overrightarrow{\mathcal{I}}}^{o,<<}_i)$, $\mathcal{S}$ contains a representative for $\mathcal{\overrightarrow{\mathcal{I}}}^{o,<<}_i$ unless this interval is redundant for $\mathcal{S}$.
Let us show that it is not: If $\mathcal{\overrightarrow{\mathcal{I}}}^{o,<<}_i$ was redundant for $\mathcal{S}$, there would exist some $k'$ such that $\mathcal{\overrightarrow{\mathcal{I}}}^{o,<<}_i\in \mathcal{I}^{s,<}_{k'}$, $\mathcal{\overrightarrow{\mathcal{I}}}^{o,<<}_i$ contains $x_{k'}$ and $\mathcal{S}$ contains a tabular subsolution with index $k'$. Let $j=\mathrm{ind}(\mathcal{\overrightarrow{\mathcal{I}}}^{o,<<}_i)$ (i.e.,  $I_j=\mathcal{\overrightarrow{\mathcal{I}}}^{o,<<}_i$). We know that $x_j\leq x_i-(1+q)$ and, since $I_j$ contains $x_{k'}$, $x_{k'}\leq x_j+1\leq x_i-q<x_i-(1-q)$. Since $\mathcal{I}^{o,<}_{i}\cup \mathcal{I}^{o,f}_{i}=\emptyset$, we obtain  $x_{k'}\leq x_i-(1+q)$. Thus, there would be an interval ($I_{k'}$) of opposite color from $I_i$ with $x_{k'}\leq x_i-(1+q)$ but to the right of $\mathcal{\overrightarrow{\mathcal{I}}}^{o,<<}_i$, a contradiction. In consequence, $\mathcal{\overrightarrow{\mathcal{I}}}^{o,<<}_i$ is not redundant for $\mathcal{S}$ and $\mathcal{S}$ contains its representative. Additionally, $\mathcal{\tilde{S}}\cup M[i]$ also does because the position of the representative is compatible with $M[i]$. By Corollary~\ref{cor:seq-cons-2}, in $\mathcal{\tilde{S}}\cup M[i]$, $M[i]$ satisfies (C2) except for intervals that are redundant for the sequence.

We conclude that $\mathcal{\tilde{S}}\cup M[i]$ forms a valid sequence of tabular subsolutions. Its indices can be obtained by taking the sequence of indices for $\mathcal{\tilde{S}}$ and adding $i$ at the end.

Finally, let us discuss the case when $\mathcal{\tilde{S}}=\emptyset$. From the previous paragraphs it is clear that this can only happen when $\mathcal{I}^{o,<<}_{i}=\emptyset$ and $\mathcal{I}^{s,<<}_i=\emptyset$. In this case of the proof, this necessarily implies that $\mathcal{I}^{s,<}_i\neq \emptyset$ (and recall that by assumption $\mathcal{I}^{o,<}_i=\emptyset$). All of this implies that the only intervals with index smaller than $i$ are intervals of red color containing $x_i$. Thus, $M[i]$ forms a valid sequence of tabular subsolutions because, even though those intervals do not have a representative in $M[i]$, they are redundant for $M[i]$.
\end{proof}

\setlength{\textfloatsep}{0pt}
\begin{algorithm}
\begin{algorithmic}[1]
\vspace{0.25cm}
\Require  $\cal I$ 
\Ensure  a solution, if it exists; otherwise, ``there is no solution"
\For{$i=1,2,\ldots,n$}
 \State compute $\mathcal{\overrightarrow{\mathcal{I}}}^{o,<}_i$, $\mathcal{\overrightarrow{\mathcal{I}}}^{s,<}_i$, $\mathcal{\overrightarrow{\mathcal{I}}}^{o,<<}_i$ \label{line:ini1}
\EndFor
\For{$i=1,2,\ldots,n$} \label{line:secondloop-1}
\State compute $M[i]$, $\max^r[i]$, $\max^b[i]$
\State $\mathrm{append}[i]:=F$ 
\State $p[i]:=\emptyset$ \Comment{$p[i]$ stores an index $k$ such that the predecessor of $M[i]$ in a valid sequence of tabular subsolutions is a tabular subsolution starting at $I_k$}
\State $last^r[i]:=-1$, $last^b[i]:=-1$ \Comment{$last^r[i]$ and $last^b[i]$ store the index of the rightmost red and blue interval with representative in a valid sequence of tabular subsolutions finishing with $M[i]$} \label{line:secondloop-2}
\EndFor
\For{$i=1,2,\ldots,n$ such that $M[i]\neq \emptyset$}
 \If {$\mathcal{I}^{s,<}_i=\mathcal{I}^{o,<<}_i=\mathcal{I}^{o,<}_i=\emptyset$} \label{line:leftmost-1}
 \State $\mathrm{append}[i]:=T$ \label{line:firstT}
 \State $last^r[i]:=\max^r[i]$, $last^b[i]:=\max^b[i]$ \label{line:leftmost-2}
 \ElsIf{$\mathcal{I}^{o,<}_i\neq \emptyset$} \label{line:notempty-1}
        \If{there exists $k\in\{1,2,\ldots,i-1\}$ such that $\mathrm{append}[k]=T$ and $M[k]$ contains a representative for $\mathcal{\overrightarrow{\mathcal{I}}}^{o,<}_i$ at a position $\leq x_i-q$} \label{line:secondT}
        \State $\mathrm{append}[i]:=T$ \label{line:ext-1}
        \State $p[i]:=k$
        \State compute $last^r[i]$, $last^b[i]$ \label{line:last-1}
        \EndIf
   \Else \label{line:last-case}
    \State \Comment{Store in $req^r[i]$ and $req^b[i]$ the highest index of the intervals of red and blue color whose representatives must be contained in a predecessor of $M[i]$}
        \If{$I_i$ is red}
            \State $req^r[i]=\mathrm{ind}(\mathcal{\overrightarrow{\mathcal{I}}}^{s,<}_i)$, $req^b[i]=\mathrm{ind}(\mathcal{\overrightarrow{\mathcal{I}}}^{o,<<}_i)$ 
        \Else
            \State $req^r[i]=\mathrm{ind}(\mathcal{\overrightarrow{\mathcal{I}}}^{o,<<}_i)$, $req^b[i]=\mathrm{ind}(\mathcal{\overrightarrow{\mathcal{I}}}^{s,<}_i)$
        \EndIf
      \If{there exists $k\in\{1,2,\ldots,i-1\}$ such that $\mathrm{append}[k]=T$, $last^r[k]\geq req^r[i]$ and $last^b[k]\geq req^b[i]$} \label{line:thirdT}
        \State $\mathrm{append}[i]:=T$ \label{line:ext-2}
        \State $p[i]:=k$
        \State compute $last^r[i]$, $last^b[i]$ \label{line:last-2}
          \EndIf
  \EndIf
\If{$\mathrm{append}[i]=T$ and $M[i]$ contains a representative for $I_n$} \label{line:imp}
  \State reconstruct and return a solution based on the reverse of the sequence $i,p[i],p[p[i]],\ldots$ \label{line:recons}
  \State finish 
\EndIf
\EndFor
\State return ``there is no solution" \label{line:final}
\end{algorithmic}
\caption{Decision problem for $q\in(1/2,3/4]$}
\label{alg:alg-main}
\end{algorithm}

The previous lemma helps us prove that the algorithm is correct:

\begin{lemma} \label{lem:corr-alg}
Algorithm~\ref{alg:alg-main} is correct.
\end{lemma}

\begin{proof}
If the algorithm answers ``yes", it is because, for some $i$, it has set $\mathrm{append}[i]=T$ and $M[i]$ contains a representative for $I_n$. By Lemma~\ref{lem:corr-lemma}, in this case there exists a valid sequence $\mathcal{S}$ of tabular subsolutions finishing with $M[i]$, and $M[i]$ is a maximal tabular subsolution containing a representative for $I_n$. If $\mathcal{S}$ is almost valid, as argued before it can easily be augmented to a super valid sequence. Lemma~\ref{lem:alm-sol} then implies that this super valid sequence of tabular subsolutions is a solution.

Next, we argue that it is not possible that there exists a solution and the algorithm answers ``no". If there exists a solution, by Observation~\ref{obs:left-sol-enough},  Lemma~\ref{lem:prun-tab} and Lemma~\ref{lem:conditions}, there exists a solution $\mathcal{S}$  with the shape of a sequence of non-overlaping tabular subsolutions satisfying conditions (C0)-(C3). Let $i_1,i_2,\ldots,i_m$ be the sequence of indices of the tabular subsolutions of $\mathcal{S}$, again with the convention that a new tabular subsolution begins at every interval $I_i$ such that, in $\mathcal{S}$, $r_i=x_i$ and $r_i$ is not
at a distance $q$ to the right of a representative of opposite color. In the next paragraphs, we prove by induction that, for each of these indices, the algorithm sets $\mathrm{append}[\cdot]:=T$. Since $\mathcal{S}$  contains a representative for $I_n$, one of these indices will then satisfy both conditions of line~\ref{line:imp} and thus the algorithm will answer ``yes".

Lemma~\ref{lem:conditions} guarantees that $\mathcal{I}^{s,<}_{i_1}=\mathcal{I}^{o,<<}_{i_1}=\mathcal{I}^{o,<}_{i_1}=\emptyset$. Clearly, $M[i_1]\neq \emptyset$. Thus, in line~\ref{line:firstT} the algorithm sets $\mathrm{append}[i_1]:=T$. Let us next assume that the algorithm sets $\mathrm{append}[\cdot]:=T$ for $i_1,i_2,\ldots,i_j$. Since $i_{j+1}$ is greater than all of these indices, $\mathrm{append}[\cdot]=T$ already (for $i_1,i_2,\ldots,i_j$) when $i_{j+1}$ is processed. However, we notice that the sequence of indices for $i_j$ (and the other indices) found by the algorithm (i.e., the reverse of $p[i_j],p[p[i_j]]\ldots$) might differ from the sequence $i_1,i_2,\ldots,i_j$. If $\mathcal{I}^{s,<}_{i_{j+1}}=\mathcal{I}^{o,<<}_{i_{j+1}}=\mathcal{I}^{o,<}_{i_{j+1}}=\emptyset$, the algorithm sets $\mathrm{append}[i_{j+1}]:=T$.
It remains to consider two more cases.

Suppose first that $\mathcal{I}^{o,<}_{i_{j+1}}\neq \emptyset$. By Lemma~\ref{lem:conditions}, in $\mathcal{S}$ all intervals in $\mathcal{I}^{o,<}_{i_{j+1}}$ have representative in $L[i_{j+1}]$ at a position $\leq x_{i_{j+1}}-q$. In particular, this is the case for $\mathcal{\overrightarrow{\mathcal{I}}}^{o,<}_{i_{j+1}}$. Let $i_{j'}\in \{i_1,i_2,\ldots,i_j\}$ be the index of the tabular subsolution of $\mathcal{S}$ containing the representative of $\mathcal{\overrightarrow{\mathcal{I}}}^{o,<}_{i_{j+1}}$. By induction hypothesis, $\mathrm{append}[i_{j'}]=T$. Additionally, $M[i_{j'}]$ contains a representative for $\mathcal{\overrightarrow{\mathcal{I}}}^{o,<}_{i_{j+1}}$ at a position $\leq x_{i_{j+1}}-q$. Thus, $i_{j'}$ satisfies all conditions in line~\ref{line:secondT} and hence the algorithm sets $\mathrm{append}[i_{j+1}]:=T$.

Finally, let us suppose that $\mathcal{I}^{o,<}_{i_{j+1}}= \emptyset$ and at least one of $\mathcal{I}^{s,<}_{i_{j+1}}$ or $\mathcal{I}^{o,<<}_{i_{j+1}}$ is non-empty. Without loss of generality, assume that $I_{i_{j+1}}$ is red. By Lemma~\ref{lem:conditions}, in $\mathcal{S}$ all intervals in $\mathcal{I}^{s,<}_{i_{j+1}}\cup \mathcal{I}^{o,<<}_{i_{j+1}}$ have representative in $L[i_{j+1}]$. In particular, this is the case for $\mathcal{\overrightarrow{\mathcal{I}}}^{s,<}_{i_{j+1}}$ and $\mathcal{\overrightarrow{\mathcal{I}}}^{o,<<}_{i_{j+1}}$. Let $i_{j'}\in \{i_1,i_2,\ldots,i_j\}$ be the smallest index such that the tabular subsolutions of $\mathcal{S}$ with index at most $i_{j'}$ contain representatives for $\mathcal{\overrightarrow{\mathcal{I}}}^{s,<}_{i_{j+1}}$ and $\mathcal{\overrightarrow{\mathcal{I}}}^{o,<<}_{i_{j+1}}$. Let $\mathcal{S'}$ be the sequence of tabular subsolutions obtained from $\mathcal{S}$ by removing the tabular subsolutions with index greater than $i_{j'}$, and replacing the tabular subsolution of index $i_{j'}$ by $M[i_{j'}]$. Clearly, $\mathcal{S'}$ contains representatives for $\mathcal{\overrightarrow{\mathcal{I}}}^{s,<}_{i_{j+1}}$ and $\mathcal{\overrightarrow{\mathcal{I}}}^{o,<<}_{i_{j+1}}$. Additionally, by Lemma~\ref{lem:conditions}, $\mathcal{S'}$ is a super valid sequence of tabular subsolutions. On the other hand, by induction hypothesis, $\mathrm{append}[i_{j'}]=T$. Thus, by Lemma~\ref{lem:corr-lemma}, the algorithm has found a valid sequence $\mathcal{S''}$ of tabular subsolutions finishing with $M[i_{j'}]$. 
Lemma~\ref{lem:seq-cons-2} determines the rightmost interval of each color represented in a valid sequence of tabular subsolutions finishing with $M[i_{j'}]$, unless this interval is redundant for the corresponding sequence. If it were, there would be an interval of the same color with higher index that has a representative in the sequence, contradicting Lemma~\ref{lem:seq-cons-2}. 
Thus, the rightmost intervals of each color represented in $\mathcal{S'}$ and $\mathcal{S''}$ are the same.
This implies that $last^r[i_{j'}]\geq \mathrm{ind}(\mathcal{\overrightarrow{\mathcal{I}}}^{s,<}_{i_{j+1}})$ and $last^b[i_{j'}]\geq \mathrm{ind}(\mathcal{\overrightarrow{\mathcal{I}}}^{o,<<}_{i_{j+1}})$.
Therefore, the algorithm finds a suitable $k$ in line~\ref{line:thirdT} and sets $\mathrm{append}[i_{j+1}]:=T$.
\end{proof}

\subsection{Efficient implementation of Algorithm~\ref{alg:alg-main}} 

The goal of this subsection is to describe an implementation of the algorithm running in $O(n\log n)$ time. Such an implementation follows the basic scheme of the pseudocode given in Algorithm~\ref{alg:alg-main}, but includes additional subroutines and variables that do not appear in the pseudocode.


The first two lines of the algorithm take in total $O(n \log n)$ time:
For each color, we create an array containing the intervals of that color sorted from left to right. Using binary search in this array, for a fixed $i$, each of the intervals in line~\ref{line:ini1} can be found in $O(\log n)$ time. 

The second loop of the algorithm (lines~\ref{line:secondloop-1}-\ref{line:secondloop-2}) contains some initialization  plus the computation of $M[i]$, $\max^r[i]$ and $\max^b[i]$.

\subsubsection{Computing $M[i]$, $\max^r[i]$ and $\max^b[i]$} In the remainder of the paper, we use $\bar{I}$ to denote the minimum segment containing all the intervals in $\cal{I}$.

As argued before, if $\mathcal{I}^{o,f}_i\neq \emptyset$, then $M[i]=\emptyset$. In $O(\log n)$ time we can determine whether $\mathcal{I}^{o,f}_i\neq \emptyset$ and, in this case, already return $M[i]=\emptyset$.

For the remaining of this part, let us assume that $\mathcal{I}^{o,f}_i= \emptyset$.
Strictly speaking, it is enough to be able to compute the position of the rightmost representatives of $M[i]$. Suppose that they have color $c$ and they are at $x_i+kq$. By Observation~\ref{obs:end-Mi}, either: (i) there is no interval of color $\bar{c}$ that contains the point $x_i+(k+1)q$; or
(ii) there is an interval of color $\bar{c}$ that contains the point $x_i+(k+1)q$, but the position $x_i+(k+1)q$ is forbidden for color $\bar{c}$.

To take care of (i), we describe a preprocessing of the intervals that allows to detect in $O(\log n)$ time the first time that there are no intervals of the appropriate color containing some point of the form $x_i+(k+1)q$. For each of the two colors, we preprocess the intervals of that color as follows.


Suppose that we are dealing with color red.
We divide $\bar{I}$ into the \emph{red} region, containing all points of $\bar{I}$ covered by at least one of the red intervals, and the \emph{white} region, containing the rest. Afterwards, we ``cut" $\bar{I}$ into portions of length $2q$ that are closed at the left endpoint and open at the right endpoint. The last portion is closed also at the right endpoint and, if its length is smaller than $2q$, we extend it until it has length $2q$ and we color the added portion in white. Notice that the number of such portions is at most $\frac{3n}{2}$.\footnote{The assumptions that $x_1=0$ and that no two consecutive intervals leave a gap of length $q$ or greater imply that the rightmost interval ends at most at $n+(n-1)q$.
 In combination with $q>\frac12$, this implies the bound of $\frac32 n$ portions.
} We then translate these portions (or segments) of length $2q$ in the plane as follows: All portions have their left endpoint on the line $x=0$. Additionally, the portion $[0,2q)$ is placed at $y=0$, the portion $[2q,4q)$ is placed at $y=2$, the portion $[4q,6q)$ is placed at $y=4$, and so on (see Fig.~\ref{fig:algorithm} for an example). 
Finally, we ignore the red regions and focus on the white regions, which become a set of (open) segments in the plane, called $T^r$.
We then preprocess $T^r$ in $O(n\log n)$ time and $O(n)$ space to answer ray-shooting queries into upwards direction in $O(\log n)$ time per query~\cite{Berg-ray-shoot}. We observe that the standard data structure to answer ray-shooting queries into upwards direction takes closed (rather than open) segments as input, so one needs to perform minor adjustments to it.

\begin{figure}[tb]
    \centering
    \includegraphics[scale=0.8]{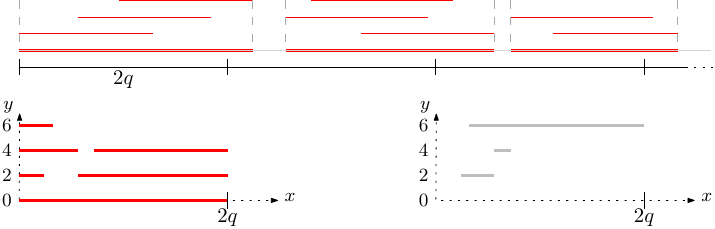}
    \caption{Construction of $T^r$: on top, the red intervals of the input and the red region; on the bottom left, the portions of length $2q$ have been cut and translated; on the bottom right, the final set of segments $T^r$.}
    \label{fig:algorithm}
\end{figure}

Let $2jq+\delta$, where $j$ is a non-negative integer and $\delta\in[0,2q)$, be the position of the leftmost red representative of $M[i]$ (if any). Then the red representatives of $M[i]$ are at $2jq+\delta$, $2(j+1)q+\delta$, $2(j+2)q+\delta\ldots$ Consider the following query: If we shoot a vertical ray with origin at $(\delta,2j+1)$ in upwards direction, when is the first time that it intersects a segment of $T^r$? 
If the answer to this query is the point $(\delta,2k)$ ($k>j$), then the first time that we cannot place the next red representative---as there is no red interval covering that position---happens at $2kq+\delta$. 

By constructing an analogous data structure for the blue intervals and performing the appropriate query, we obtain the first time (if any) that a blue representative cannot be placed because there is no blue interval intersecting that position. 

We next take care of (ii), that is, we show how to detect the first time that a red or blue representative cannot be placed because that position is forbidden. Again, we will preprocess the intervals so we can detect it in $O(\log n)$ time.

The strategy is the same as in the previous case: In the preprocessing phase, we mark all portions of $\bar{I}$ corresponding to the forbidden positions for one of the colors. In detail, given a red interval $I_i$, in a feasible representation there cannot be any blue representative in the interval $(x_i+(1-q),x_i+q)$. Thus, in $\bar{I}$, we mark $(x_i+(1-q),x_i+q)$ in color dark blue. Repeating this for all red intervals, we obtain the \emph{dark blue} region of $\bar{I}$; this clearly corresponds to the portions where blue representatives cannot be placed.

Analogously, we compute the \emph{dark red} region of $\bar{I}$. As in the previous case, for each color we cut $\bar{I}$ into portions of length $2q$, and we distribute the portions into a vertical column of segments of length $2q$. The starting point of a maximal tabular solution $M[i]$ tells us from which position we shoot a vertical ray upwards and compute the first intersection with a dark red (or dark blue) region. For each color, we obtain an upper bound for the rightmost point of $M[i]$. 

The minimum of the four obtained upper bounds gives us the rightmost point of $M[i]$. 
Let us explain how to compute $\max^r[i]$ and $\max^b[i]$. Suppose, e.g., that the rightmost representatives of $M[i]$ are red and are at position $x$. Then, in $O(\log n)$ time  we can find the rightmost red interval containing $x$. The index of this interval is precisely $\max^r[i]$. Analogously, $\max^b[i]$ corresponds to the index of the rightmost blue interval containing $x-q$, provided that $x-q\geq x_i$. If $x-q< x_i$, $M[i]$ does not contain any blue representative and we set $\max^b[i]=-1$.

\subsubsection{Case $\mathcal{I}^{o,<}_i\neq \emptyset$}

The next challenging part occurs in line~\ref{line:secondT} of the algorithm.
We need to figure out whether there exists  a $k\in\{1,2,\ldots,i-1\}$ such that $\mathrm{append}[k]=T$ and $M[k]$ contains a representative for $\mathcal{\overrightarrow{\mathcal{I}}}^{o,<}_i$ at a position $\leq x_i-q$.
To this end, for each color, we use a preprocessing similar to the one we have just described.

Let us consider a red interval $I_i$ such that $M[i]\neq \emptyset$ and $\mathcal{I}^{o,<}_i\neq \emptyset$. 
Then $\mathcal{\overrightarrow{\mathcal{I}}}^{o,<}_i$ is a blue interval, and the positions for its representative that allow for a valid sequence of tabular subsolutions to be combined with $M[i]$ are those in the segment $\mathcal{\overrightarrow{\mathcal{I}}}^{o,<}_i\cap (-\infty,x_i-q]$, which we call the \emph{good segment} of $\mathcal{\overrightarrow{\mathcal{I}}}^{o,<}_i$ and which, by definition of $\mathcal{I}^{o,<}_i$, is non-empty.   Each red interval with $M[i]\neq \emptyset$ and $\mathcal{I}^{o,<}_i\neq \emptyset$ gives us a good segment, and we take the collection $S_1$ of all such segments of positive length in $\bar{I}$.\footnote{It is not needed to add the good segment of $\mathcal{\overrightarrow{\mathcal{I}}}^{o,<}_i$ to $S_1$ if it consists of exactly one point $x$ (with $x=x_i-q$) for the following reason: According to the way in which we divide a leftmost solution into a sequence of non-overlaping tabular subsolutions, a leftmost solution with the representative of $\mathcal{\overrightarrow{\mathcal{I}}}^{o,<}_i$ at $x$ and $r_i=x_i$ does not contain a tabular subsolution starting at $I_i$ because $r_i$ is
at a distance $q$ to the right of $x$, which is a representative of opposite color. Hence, in this case it is not needed to determine if $M[i]$ is appendable and we can simply set $\mathrm{append}[i]=F$. In other words, concatenating $M[i]$ with an appropriate portion of some appropriate $M[k]$ gives the same as simply considering the entire $M[k]$.} 


As before, we cut $\bar{I}$ 
 into portions of size $2q$, and we place the portions in a vertical column of segments, following, in principle, the same  rules as before. 
 Again, the idea is to preprocess the segments in $S_1$ for efficient ray shooting queries.
 However, the segments in $S_1$ may overlap, which is an issue for standard ray-shooting data structures. To overcome this, we remove any overlap by displacing, by a positive amount smaller than $1$, the $y$-coordinate of all segments in each portion, so that no two segments have the same $y$-coordinate.
  The resulting collection of horizontal segments, all at different heights, is called $G^b$. 
We compute $G^r$ in an analogous way.  
Then, in $O(n\log n)$ time and $O(n)$ space, we preprocess each of the sets $G^b$ and $G^r$ to answer ray shooting queries into upwards direction in $O(\log n)$ time. This time we use a dynamic data structure that allows to insert and delete segments in $O(\log n)$ time~\cite{GiyoraK09}.

During the main loop of the algorithm, every time that $\mathrm{append}[k]$ is set to $T$ for some $k$, for each of the two colors we proceed as follows. Take, e.g., color red. During the computation of $M[k]$, we have obtained the position of its rightmost representatives. This information, together with the position and color of $x_k$, is enough to derive the positions of all red representatives of $M[k]$. These positions can be expressed as a list of the shape $2jq+\delta,2(j+1)q+\delta,\ldots,2j'q+\delta$, where $j,j'$ are non-negative integers, $j'\geq j$ and $\delta \in [0,2q)$ (or there is no red representative in $M[k]$, in which case there is nothing to do). We then go to $G^r$ and perform a ray shooting query with a vertical segment starting at $(\delta,2j)$ and going upwards. If the first intersection (if any) with a segment of $G^r$ occurs at a point $(\delta,y)$ with $y\geq 2j'+1$, we stop.

Otherwise, the first intersection happens at a point $x$ that belongs to some segment $s_1\in S_1$.
It corresponds to the good segment of some interval $\mathcal{\overrightarrow{\mathcal{I}}}^{o,<}_i$, where $I_i$ is blue and $i>k$.
The fact that the vertical segment with endpoints $(\delta,2j)$, $(\delta,2j'+1)$  intersects the good segment of $\mathcal{\overrightarrow{\mathcal{I}}}^{o,<}_i$ implies that $M[k]$ contains a representative for $\mathcal{\overrightarrow{\mathcal{I}}}^{o,<}_i$ at a position $\leq x_i-q$. We register this information in a separate list $p'$.
In particular, we store $p'[i]:=k$. 
Then, we remove $s_1$ from the ray shooting data structure $G^r$. 
%
Afterwards, we repeat the ray shooting query with a vertical segment starting at $(\delta,2j)$ and going upwards, in case there are still intersections with segments in $S_1$ above $x$ and below $(\delta,2j'+1)$.

In summary, in order to perform the test in line~\ref{line:secondT} efficiently, we proceed as follows: Every time that $\mathrm{append}[k]$ is set to $T$ for some $k$, we perform the described ray shooting queries. This might result in setting $p'[i]:=k$ for several $i>k$. The test in line~\ref{line:secondT} is then done by checking $p'[i]$: If $p'[i]=k$, this value satisfies all conditions in line~\ref{line:secondT}. On the other hand, if $p'[i]=\emptyset$, there does not exist any $k\in\{1,2,\ldots,i-1\}$ such that $\mathrm{append}[k]:=T$ and $M[k]$ contains a representative for $\mathcal{\overrightarrow{\mathcal{I}}}^{o,<}_i$ at a position $\leq x_i-q$.
We point out that, to keep the pseudocode in Algorithm~\ref{alg:alg-main} as simple and understandable as possible, these additional steps do not appear there.

Let us analyze the total running time. Initially, $G^b$ and $G^r$ are built and preprocessed in $O(n\log n)$ time. 
Then, every time that $\mathrm{append}[k]$ is set to $T$ for some $k$, for each of the two colors, the algorithm performs one or more ray shooting queries. In total, it performs exactly one ``unsuccessful" query (when there is no intersection with $G^r$ or the first intersection with a segment of $G^r$ occurs at a point $(\delta,y)$ with $y\geq 2j'+1$), which is charged to $k$. On the other hand, every ``successful" ray shooting query (intersecting a segment in $G^r$ ``within the good range") and every segment in $S_1$ that contains the intersected segment in $G^r$ cause a constant number of updates in $G^r$, each of which takes $O(\log n)$ time. Notice that the cost can be charged to the segment $s_1$ in $S_1$ (or alternatively, the interval $I_i$ such that $s_1$ is the good segment of $\mathcal{\overrightarrow{\mathcal{I}}}^{o,<}_i$). The reason is that afterwards $s_1$ is removed from $G^r$, so no more operations are charged to it in future steps of the algorithm. Consequently, the total cost of all the involved operations is $O(n\log n)$.

\subsubsection{Computing $last^r[i]$ and $last^b[i]$ in lines~\ref{line:last-1} and~\ref{line:last-2}}

For this step, it would be tempting to use $last^r[i]:=\max^r[i]$ and $last^b[i]:=\max^b[i]$, but this is not necessarily true if $M[i]$ contains representatives of only one color. A correct alternative option is to use Lemma~\ref{lem:seq-cons-2}: If the rightmost representative of $M[i]$ is at position $x$ and has color $c$, any valid sequence of tabular subsolutions finishing with $M[i]$ contains representatives for all intervals $I_j$ of color $c$ such that $x_j\leq x$ and for all intervals $I_b$ of color $\bar{c}$ such that $x_b< x+q-1$, except for those that are redundant for the sequence. For each color, we can find in $O(\log n)$ time the index of the rightmost interval satisfying the appropriate inequality. Notice that it follows from the definition of redundant interval that such rightmost intervals cannot be redundant. Thus, the two obtained indices correspond to $last^r[i]$ and $last^b[i]$. Consequently, the total cost of all the executions of lines~\ref{line:last-1} and~\ref{line:last-2} is $O(n\log n)$.

\subsubsection{Case $\mathcal{I}^{o,<}_i= \emptyset$, and at least one of $\mathcal{I}^{s,<}_{i_{j+1}}$ or $\mathcal{I}^{o,<<}_{i_{j+1}}$ is non-empty (line~\ref{line:last-case})}

For this case it is important to be able to perform the test in line~\ref{line:thirdT} efficiently;
that is, to decide if there is a previous subsolution with $last^r[k]\geq req^r[i]$ and $last^b[k]\geq req^b[i]$.
To this end, we use a data structure for dynamic orthogonal range searching in 2D. 

We start with an empty data structure $S$.
During the main loop of the algorithm, every time that $\mathrm{append}[k]$ is set to $T$ for some $k$, we add the point $(last^r[k],last^b[k])$ to $S$. The test in line~\ref{line:thirdT} is then equivalent to the existence of some point $(x,y)\in S$ such that $x\geq req^r[i]$ and $y\geq req^b[i]$.
This is a 2-sided orthogonal range query for points in the plane.
Many data structures have been proposed that can answer such a query in $O(\log n)$ time, with $O(\log n)$ insertion time,  using $O(n)$ space in total.
For example, we can do this with a priority search tree~\cite{McCreight85}.
Note, however, that much more efficient and sophisticated data structures exist for this problem~\cite{ChanT18}.

In $S$, we perform at most $n$ insertions and at most $n$ queries. Thus, the total cost associated to line~\ref{line:thirdT} and to all operations involving $S$ is again $O(n\log n)$.

\subsubsection{Reconstructing a solution (line~\ref{line:recons})}

The procedure to reconstruct a solution based on the information obtained by the algorithm has already been outlined in the proof of Lemma~\ref{lem:corr-lemma}. 

Let $i$ be such that $\mathrm{append}[i]:=T$ and $M[i]$ contains a representative for $I_n$. We consider the sequence $i,p[i],p[p[i]],\ldots$ and, for each index, we must determine which of the tabular subsolutions with that index (if any) is included in the valid sequence. For $i$, we take the complete $M[i]$. If $p[i]=\emptyset$, we are done. If $\mathrm{append}[i]$ has been set to $T$ in line~\ref{line:ext-1}, for $p[i]$ we take the longest compatible tabular subsolution of $M[p[i]]$ with respect to $M[i]$. We move to index $p[i]$, and we repeat. If $\mathrm{append}[i]$ has been set to $T$ in line~\ref{line:ext-2}, we consider the longest compatible tabular subsolution of $M[p[i]]$ with respect to $M[i]$. If it is non-empty, we take it, we move to index $p[i]$, and we repeat. If it is empty, we traverse the sequence $p[p[i]],p[p[p[i]]]\ldots$ until we find an element $k'$ such that the longest compatible tabular subsolution of $M[k']$ with respect to $M[i]$ is non-empty. For $p[i]$ and the other indices traversed before $k'$, we take the empty tabular subsolution with that index; for $k'$, we take the longest compatible tabular subsolution of $M[k']$ with respect to $M[i]$. We move to index $k'$, and we repeat.

After the previous step, reconstructing the valid final sequence of tabular subsolutions is trivial because we know the left-to-right order of the tabular subsolutions and the precise length of each of them. Recall that this sequence might omit representatives of some redundant intervals, and to obtain a solution, we have to add them. 

All in all, the total cost of line~\ref{line:recons} is $O(n)$.

\subsubsection{Putting everything together}

We have seen that all the steps of the algorithm can be performed in $O(n\log n)$ time. We obtain the first main result of the paper:

\begin{theorem} \label{thm:dec-diff}
For $k=2$, the decision problem for $q\in \left(\frac 1 2, 
\frac 3 4 \right]$ can be solved in $O(n \log n)$ time.
\end{theorem}

\section{Optimization problem}

In this section we present an algorithm for the optimization problem. Recall that $q^*$ is the length of a \mcsi\ in a realization solving the problem. 

\begin{lemma} \label{lem:cand-q}
There exist $i,j$ such that $q^*=\frac{x_j+1-x_i}{m}$, where $m\in \mathbb{N}$.
\end{lemma}

\begin{proof}
Let us take a realization ${\cal P}^*$ solving the problem. We transform it into a leftmost solution without changing the length of the \mcsi(s).
Recall that, since there are only two colors, $q^*$ is simply the minimum distance between two consecutive representatives of different color in ${\cal P}^*$.
Let $r_h$ and $r_{\ell}$ be two representatives of different color with  distance $q^*$. 
Observe that $r_h$ and $r_{\ell}$ must be part of the same tabular subsolution ${\cal{T}}$, otherwise they would be further apart than  $q^*$.\footnote{Here we assume that we decompose the leftmost solution into a sequence of non-overlaping tabular subsolutions as in the proof of Lemma~\ref{lem:prun-tab}.}

Let $r_i$ be the leftmost representative in ${\cal{T}}$ and $r_k$ be a rightmost representative in ${\cal{T}}$.
By definition of tabular subsolution, $r_i$ is at the left endpoint of $I_i$.
Next we argue that some representative of ${\cal{T}}$ must be at the right endpoint of its interval.
Indeed, if this was not the case, it would be possible to move $r_k$ by a distance $\varepsilon$ to the right (this would be possible, because the tabular subsolution ends at $r_k$, thus there is no representative of opposite color at distance $q^*$ to the right).
In turn,  moving  $r_k$ slightly to the right would allow moving all representatives in the subsolution slightly to the right (this would be possible, because by assumption no representative of ${\cal{T}}$ is at the right endpoint of its interval), increasing all distances between consecutive representatives of different colors in this subsolution. 
Repeating the argument for all tabular subsolutions
containing a pair of representatives of different color at  distance $q^*$, 
we would arrive at a solution with separation larger than $q^*$, a contradiction.

It follows that the representative of some $I_j$ represented in ${\cal{T}}$ is at $x_j+1$.
Now, since in the tabular subsolution the representatives of different colors alternate at a distance of exactly $q^*$, it follows that if there are in total $m$ alternations between $r_i$ and $r_j$, then $q^* = \frac{x_j+1-x_i}{m}$.
\end{proof}


How many values of the form $\frac{x_j+1-x_i}{m}$ are there? Recall that we have assumed that $x_1=0$. Since $q^* \leq 2$, we can also assume that no two consecutive
intervals leave a gap of length 2 or greater (recall that for the decision problem with a fixed $q$, we argued that we could assume that no two consecutive
intervals leave a gap of length $q$ or greater). This implies that $x_n+1<3n$. Recall also that $\frac{1}{2}\leq q^*$.

For fixed $i,j$, $$\frac{1}{2}\leq \frac{x_j+1-x_i}{m} \leq 2, \mathrm{ with }\ m\in \mathbb{N},$$ if and only if $$\frac{x_j+1-x_i}{2}\leq m \leq 2 (x_j+1-x_i), \mathrm{ with }\ m\in \mathbb{N}.$$ Since $x_j+1-x_i< 3n$, there are less than $6n$ values of $m\in \mathbb{N}$ satisfying $\frac{1}{2}\leq \frac{x_j+1-x_i}{m} \leq 2$. Thus, the total number of candidates for $q^*$ is $O(n^3)$. Listing all such candidates, sorting them and running a binary search on them using the algorithms for the decision version of the problem yields an algorithm running in $O(n^3\log n)$ time. 
Next we present a faster approach.

Let us start by running a binary search on the values of type $\frac{x_n+1-x_1}{m}$, where $m\in \mathbb{N}$ and $\frac{1}{2}\leq \frac{x_n+1-x_1}{m}\leq 2$. In $O(n \log^2 n)$ time we can find some $m_0\in \mathbb{N}$ such that $\frac{x_n+1-x_1}{m_0+1}\leq q^* < \frac{x_n+1-x_1}{m_0}$. 
This already gives us a feasible solution for $q=\frac{x_n+1-x_1}{m_0+1}$.
Thus for any other $i,j$, we are then only interested in the values satisfying

\begin{equation}
\label{eq:ijvalues}
\frac{x_n+1-x_1}{m_0+1}< \frac{x_j+1-x_i}{m} < \frac{x_n+1-x_1}{m_0},\end{equation}
with $m\in \mathbb{N}$. 

Hence, we have $\frac{x_j+1-x_i}{x_n+1-x_1}m_0<m<\frac{x_j+1-x_i}{x_n+1-x_1}(m_0+1)$, with $m\in \mathbb{N}$. Since $x_j-x_i\leq x_n-x_1$, this equation has at most one solution. Thus, each pair $i,j$ gives at most one candidate. 

We now define $q^-=\frac{x_n+1-x_1}{m_0+1}$.
We also write $x_j+1=a_jq^-+b_j$ and $x_i=c_iq^-+d_i$, where $a_j,c_i$ are non-negative integers and $b_j,d_i\in [0,q^-)$. 
If $b_j>d_i$ and there exists some $m\in \mathbb{N}$ satisfying Equation~\ref{eq:ijvalues}, it is not difficult to see that $m=a_j-c_i$ and $\frac{x_j+1-x_i}{m}=q^-+\frac{b_j-d_i}{a_j-c_i}$. 
If $b_j=d_i$, there does not exist any $m\in \mathbb{N}$ satisfying Equation~\ref{eq:ijvalues}. Finally, if $b_j<d_i$ and there exists some $m\in \mathbb{N}$ satisfying Equation~\ref{eq:ijvalues}, we have that $m=a_j-c_i-1$ and $\frac{x_j+1-x_i}{m}=q^-+\frac{q^-+b_j-d_i}{a_j-c_i-1}$.

To be able to search among the remaining $O(n^2)$ candidates without generating them, we create a point set $S$ associated to the input intervals $\cal{I}$ and $q^-$. For every $x_i$, we add to $S$ the point $\ell_i=(d_i,c_i)$. 
For every $x_j+1$, we add to $S$ the points $r_j^1=(b_j,a_j)$ and $r_j^2=(q^-+b_j,a_j-1)$.
Observe that the slope of the segment connecting $\ell_i$ to $r_j^1$ is $\frac{a_j-c_i}{b_j-d_i}$, while that of the segment connecting $\ell_i$ to $r_j^2$ is $\frac{a_j-c_i-1}{q^-+b_j-d_i}$. Recall that, if $i,j$ give a candidate and $b_j>d_i$, this candidate is $q^-+\frac{b_j-d_i}{a_j-c_i}$. The first term does not depend on $i,j$, while the second term is the inverse of the slope of the segment connecting $\ell_i$ to $r_j^1$. If $b_j<d_i$, the candidate is $q^-+\frac{q^-+b_j-d_i}{a_j-c_i-1}$, and the second term is the inverse of the slope of the segment connecting $\ell_i$ to $r_j^2$. Hence, we can design an algorithm that runs a binary search on the set of remaining candidates through a binary search on the set of slopes defined by points in $S$.

Observe that $S$ contains $3n$ points. For every $i\in\left[1,{3n \choose 2}\right]$, we denote by $\alpha_i$ the $i$th smallest slope defined by pairs of points of $S$. We recall that the problem of computing the $k$th smallest slope defined by pairs of points of a given set is called \emph{slope selection problem}. There exist several algorithms solving this problem in $O(n\log n)$ time (e.g.~\cite{Cole-slope,Katz-slope}), even when the point set is not necessarily in general position (see, e.g.~\cite{Chazelle-slope}).

Our algorithm is described in Algorithms~\ref{alg:optimization} (main code) and~\ref{alg:optimization-sub} (code of the subroutine). In lines~\ref{line:opt-1}-\ref{line:opt-2} of Algorithm~\ref{alg:optimization}, we check two extremal cases that might already give $q^*$. In case they do not, we make the first call to the recursive subroutine \texttt{opt}. This subroutine takes as input parameters $i,j$, which are integers with $i<j$ such that the decision problem for $q^-+\frac{1}{\alpha_i}$ has a NO answer, while that for $q^-+\frac{1}{\alpha_j}$ has a YES answer.


At every new call of \texttt{opt}, we compute the candidate associated to an integer $k$ roughly equidistant from $i$ and $j$ (line~\ref{line:opt-3}). The answer of the decision problem for the associated candidate $q^-+\frac{1}{\alpha_k}$ indicates if the next call is \texttt{opt}$(i,k)$ or \texttt{opt}($k,j$) (lines~\ref{line:opt-11}-\ref{line:opt-8}). The algorithm terminates when $j=i+1$ (lines~\ref{line:opt-9}-\ref{line:opt-10}).


\begin{algorithm}
\begin{algorithmic}[1]
\vspace{0.25cm}
\Require  $\mathcal{I},S,q^-$
\Ensure  $q^*$
\State $\gamma:=$ greatest (positive) slope defined by pairs of points of $S$ that 
 are not vertically aligned \label{line:opt-1}
\If{the decision problem for $q^-+\frac{1}{\gamma}$ has a NO answer}
\State stop and return $q^-$
\Else
\State $\delta:=$ smallest among all positive slopes defined by pairs of points of $S$
\If{the decision problem for $q^-+\frac{1}{\delta}$ has a YES answer}
\State stop and return $q^-+\frac{1}{\delta}$ \label{line:opt-2}
\Else
   \State $i:=$ integer $k$ such that $\delta=\alpha_k$
   \State $j:=$ integer $k'$ such that $\gamma=\alpha_{k'}$ \label{line:opt-12}
   \State \texttt{opt}$\left(i,j\right)$
\EndIf
\EndIf
\end{algorithmic}
\caption{Algorithm \texttt{Optimization}}
\label{alg:optimization}
\end{algorithm}

\begin{algorithm}
\begin{algorithmic}[1]
\vspace{0.25cm}
\Require  $i,j$
\Ensure  $q^*$
\If{$j=i+1$} \label{line:opt-9}
\State stop and return $q^-+\frac{1}{\alpha_j}$ \label{line:opt-10}
\EndIf
\State $k:=\lfloor\frac{i+j}{2} \rfloor$ \label{line:opt-3}
  \If{the decision problem for $q^-+\frac{1}{\alpha_k}$ has a YES answer} \label{line:opt-11}
    \State \texttt{opt}($i,k$)
   \Else
     \State \texttt{opt}($k,j$) \label{line:opt-8}
  \EndIf
\end{algorithmic}
\caption{Subroutine \texttt{opt}}
\label{alg:optimization-sub}
\end{algorithm}

\begin{theorem} \label{thm:k-2-opt}
For $k=2$, the \lmcsi~ problem can be solved in $O(n\log^2 n)$ time.  
\end{theorem}

\begin{proof}
We first argue that our algorithm is correct. 

We begin with lines~\ref{line:opt-1}-\ref{line:opt-2} of Algorithm~\ref{alg:optimization}. Let $i,j$ be a pair giving a candidate equal to $q^-+\frac{b_j-d_i}{a_j-c_i}=q^-+\frac{1}{\beta_{i,j}^1}$, where $\beta_{i,j}^1>0$ is the slope of the segment connecting $\ell_i$ to $r_j^1$. If $\gamma$ is the greatest slope defined by pairs of points of $S$ that 
 are not vertically aligned and the decision problem for $q^-+\frac{1}{\gamma}$ has a NO answer, we derive that the decision problem for $q^-+\frac{1}{\beta_{i,j}^1}$ also has a NO answer because $q^-+\frac{1}{\beta_{i,j}^1}\geq q^-+\frac{1}{\gamma}$. If the candidate associated to $i,j$ is $q^-+\frac{q+b_j-d_i}{a_j-c_i-1}=q^-+\frac{1}{\beta_{i,j}^2}$, where $\beta_{i,j}^2>0$ is the slope of the segment connecting $\ell_i$ to $r_j^2$, the argument is analogous. Hence, in this case all remaining candidates correspond to NO instances and $q^*=q^-$.
 Similarly, if $\delta$ is the smallest positive slope defined by pairs of points of $S$ and the decision problem for $q^-+\frac{1}{\delta}$ has a YES answer, all remaining candidates correspond to YES instances and are not greater than $q^-+\frac{1}{\delta}$, so $q^*=q^-+\frac{1}{\delta}$.

 Otherwise, the algorithm makes the first call to the subroutine \texttt{opt}$(i,j)$. We observe that $i,j$ are integers with $i<j$ such that the decision problem for $q^-+\frac{1}{\alpha_i}$ has a NO answer, while that for $q^-+\frac{1}{\alpha_j}$ has a YES answer, so $q^-+\frac{1}{\alpha_j}\leq q^* < q^-+\frac{1}{\alpha_i}$. This property is maintained throughout all iterations. At the last iteration, $j=i+1$; since each of our candidate values $q^-+\frac{1}{\beta_{i,j}^1}$ or $q^-+\frac{1}{\beta_{i,j}^2}$ is equal to some $q^-+\frac{1}{\alpha_k}$ for some $k$ and (the last) $j$ is the smallest among all integers such that $q^-+\frac{1}{\alpha_j}$ is a YES instance (which corresponds to the greatest among all $q^-+\frac{1}{\alpha_k}$ that are a YES instance), we conclude that $q^*=q^-+\frac{1}{\alpha_j}$.

 It remains to compute the running time of the algorithm.

 The values $\gamma$, $\delta$, $i$ and $j$ of Algorithm~\ref{alg:optimization} can be computed in $O(n \log^2 n)$ time by running binary search on $\left[1,{3n \choose 2}\right]$ and computing $\alpha_k$ for the values $k$ appearing in the binary search (as mentioned earlier, every such $\alpha_k$ can be computed in $O(n\log n)$ time). 
 By Proposition~\ref{prop:dec-easy} and Theorems~\ref{thm:dec-mid} and~\ref{thm:dec-diff}, the decision problem for $q\in \left(\frac 1 2, 2 \right]$ can be solved in $O(n \log n)$ time. Hence, the total running time of lines~\ref{line:opt-1}-\ref{line:opt-12} of Algorithm~\ref{alg:optimization} is $O(n \log^2 n)$. Afterwards, the subroutine \texttt{opt} is called $O(\log n)$ times, and every time it entails a call to the slope selection problem and a call to the decision version of our problem. In consequence, the total running time of the algorithm is $O(n \log^2 n)$.
\end{proof}





\section{Open problems}
The main open problem stemming from this work is designing an efficient algorithm for $k>2$ colors.
Unfortunately, our algorithm for $k=2$ does not easily generalize to this case: Even though we can still define leftmost solutions, if we have a representative that is not at the leftmost point of its interval, there are $k-1$ possibilities for the color of the representative that is at distance $q$. Thus, the number of tabular subsolutions starting at some given $I_i$ increases, and so does, in principle, the complexity of constructing valid sequences. It follows that this technique yields a high running time in terms of $n$ and $k$ as soon as $k$ is not constant.
Therefore, a different approach might be needed to efficiently tackle instances with more than two colors.


\subsection*{Acknowledgements}

The authors would like to thank Ramesh K. Jallu and Maarten L{\"{o}}ffler for fruitful discussions.

A. Acharyya was supported by the DST-SERB grant number SRG/2022/002277. V. Keikha was supported by the CAS PPPLZ grant L100302301, and the institutional support RVO: 67985807. M. Saumell was supported by the Czech Science Foundation, grant number 23-04949X. 
R. Silveira was partially supported by grant PID2019-104129GB-I00/ MCIN/ AEI/ 10.13039/501100011033.   

\bibliographystyle{splncs04}
\bibliography{mybibliography}

\end{document}